\documentclass[reprint,twocolumn,aps,amsmath,amssymb,pra,superscriptaddress,floatfix,tightenlines]{revtex4-2}
\usepackage{amsmath, amssymb, amsthm}
\usepackage{mathtools}
\usepackage[dvipsnames]{xcolor}
\usepackage{braket}
\usepackage{comment}
\usepackage{bbm}
\usepackage{bm}
\usepackage{hyperref}
\usepackage{url}
\usepackage{color}
\usepackage{cleveref}
\usepackage{times,txfonts}
\usepackage[T1]{fontenc}

\hypersetup{
    bookmarksnumbered=true, 
    unicode=false, 
    pdfstartview={FitH}, 
    pdftitle={}, 
    pdfauthor={}, 
    pdfsubject={}, 
    pdfcreator={}, 
    pdfproducer={}, 
    pdfkeywords={}, 
    pdfnewwindow=true, 
    colorlinks=true, 
    linkcolor=NavyBlue, 
    citecolor=NavyBlue, 
    filecolor=NavyBlue, 
    urlcolor=NavyBlue 
}

\DeclareMathOperator{\tr}{Tr}

\mathchardef\mhyphen="2D

\newcommand{\Herm}{\mathrm{Herm}}

\newcommand{\ketbra}[2]{|#1\rangle\langle#2|}

\allowdisplaybreaks[4]

\theoremstyle{plain}
\newtheorem{theorem}{Theorem}
\newtheorem{lemma}[theorem]{Lemma}

\newtheorem{proposition}[theorem]{Proposition}
\theoremstyle{definition}
\newtheorem{definition}[theorem]{Definition}
\newtheorem{example}{Example}

\theoremstyle{remark}
\newtheorem{conjecture}{Conjecture}
\newtheorem{remark}[conjecture]{Remark}

\newcommand{\blue}{}

\begin{document}
\author{Kohdai Kuroiwa}
\affiliation{Institute for Quantum Computing and Department of Combinatorics and Optimization, University of Waterloo, Ontario, Canada, N2L 3G1}
\affiliation{Perimeter Institute for Theoretical Physics, Ontario, Canada, N2L 2Y5}
\author{Ryuji Takagi}
\affiliation{Department of Basic Science, The University of Tokyo, 3-8-1 Komaba, Meguro-ku, Tokyo, 153-8902, Japan}
\author{Gerardo Adesso}
\affiliation{School of Mathematical Sciences and Centre for the Mathematical and Theoretical Physics of Quantum Non-Equilibrium Systems, University of Nottingham, University Park, Nottingham, NG7 2RD, United Kingdom}
\author{Hayata Yamasaki}
\affiliation{Department of Physics, Graduate School of Science, The University of Tokyo, 7-3-1 Hongo, Bunkyo-ku, Tokyo, 113-0033, Japan}

\title{Robustness- and weight-based resource measures without convexity restriction:\\ Multi-copy witness and operational advantage in static and dynamical quantum resource theories}

\begin{abstract}
    Quantum resource theories (QRTs) provide a unified framework to analyze quantum properties as resources for achieving advantages in quantum information processing. 
    The generalized robustness and the weight of resource have been gaining increasing attention as useful resource quantifiers. 
    However, the existing analyses of these measures were restricted to the cases where convexity of the set of free states is assumed, and physically motivated resources do not necessarily satisfy this restriction.
    In this paper, we give characterizations of robustness- and weight-based measures in general QRTs without convexity restriction through two different yet related approaches. 
    On the one hand, we characterize the generalized robustness and the weight of resource by introducing a non-linear witness. 
    We show a general construction of new witness observables that detect the resourcefulness of a given state from multiple copies of the state and, using these witnesses, we provide operational interpretations of the above resource measures even without any convexity assumption. 
    On the other hand, we find that the generalized robustness and the weight of resource can also be interpreted as the worst-case maximum advantage in variants of channel-discrimination and channel-exclusion tasks, respectively, where the set of free states consists of several convex subsets corresponding to multiple restrictions.
    We further extend these results to QRTs for quantum channels and quantum instruments.
    These characterizations show that every quantum resource exhibits an advantage for the corresponding tasks, even in general QRTs without convexity assumption.
    Thus, we establish the usefulness of robustness-based and weight-based techniques beyond the conventional scope of convex QRTs, leading to a better understanding of the general structure of QRTs\@. 
\end{abstract}
\maketitle

\section{Introduction}

Quantum mechanics has opened the door to exciting possibilities in information processing, paving the way for quantum technologies to outshine classical ones in computation, communication, cryptography, and sensing~\cite{Dowling2003,Acin_2018}.  In order to transform our technological landscape and explore undiscovered benefits of quantum mechanics, fundamental quantum-mechanical phenomena need to be thoroughly characterized and leveraged as valuable assets, or {\it resources}. To facilitate this program, a framework called quantum resource theory (QRT)~\cite{Streltsov2017,Chitambar2018,Kuroiwa2020} has been established, and applied to the systematic investigation of many different quantum features, including traditionally entanglement and superposition. Common properties shared among such quantum features were also analyzed in the name of general QRTs~\cite{Streltsov2017,Chitambar2018,Horodecki2013b, Liu2017, Regula2017,Bromley2018,Anshu2018,Uola2019,Vijayan2019,Liu_ZW2019a,Takagi2019b,Takagi2019a,Fang2020,Regula2020,Sparaciari2020,Kuroiwa2020,PhysRevA.104.L020401,Regula2021,Regula2021_fundamentallimitation,Regula2021one-shot,Lami2021,Regula2022,Fang2022,Regula2022tightconstraints,Lami2023,berta2023gap}. While QRTs were initially intended to study properties of quantum states, this framework was extended to analyze dynamical properties of quantum channels~\cite{Gour2019a,Liu_ZW2019b,Liu_YC2020,Li2018,Gour_Wilde2018,Takagi2020,gour2020dynamical,yuan2020oneshot,Regula2021_fundamentallimitation,Takagi2019a,Fang2022} and quantum measurements~\cite{Heinosaari2015_NoiseRobustness,Haapasalo2015_RobustnessIncompatibility,Guerini2017_MeasurmentSimulability,Skrzypczyk2019_RobustnessMeasurement,Skrzypczyk2019_IncompatibleMeasurements,Oszmaniec2019operational,Guff2021_ResourceMeasurement}.  

Quantification of quantum resources has been of central interest to validate their usefulness, and for this reason, various resource quantifiers have been proposed both operationally and axiomatically.  
Among all of them, two promising candidates have been investigated in QRTs: the \textit{generalized robustness}~\cite{Vidal1999,Steiner2003,Harrow2003} and the \textit{weight of resource}~\cite{Lewenstein1998,Skrzypczyk2014,Pusey2015,Cavalcanti2016,Kaifeng2018,Ducuara2020,Uola2020}. 
The generalized robustness intuitively represents how resilient a given resource state is against mixing with another state. 
It was firstly defined for bipartite entanglement~\cite{Steiner2003,Harrow2003}, and generalized to QRTs characterized by a convex structure of free states. 
The generalized robustness is regarded as a valuable resource quantifier also thanks to its computability with semidefinite programming (SDP) for several representative resources such as $k$-entanglement~\cite{Bae2019}, coherence~\cite{Napoli2016}, multilevel coherence~\cite{Ringbauer2018}, asymmetry~\cite{Piani2016}, magic~\cite{Seddon2021}, and steering~\cite{Piani2015}. 
Moreover, remarkably, the generalized robustness has a direct connection to the ``usefulness'' of quantum resources in a convex QRT. 
Indeed, Ref~\cite{Takagi2019b} showed that
the generalized robustness is operationally understood as the advantage of a quantum resource state in some channel discrimination task~\cite{Kitaev_1997,Andrew2000,Acin2001,watrous_2018} in any convex QRT defined on a finite-dimensional state space, by generalizing the results for steering~\cite{Piani2015}, coherence~\cite{Napoli2016}, and asymmetry~\cite{Piani2016}. 
This general result was recently extended to infinite-dimensional convex QRTs~\cite{Regula2021,Lami2021}. 
Further, the result was extended to QRTs for channels~\cite{Takagi2019a} and measurements~\cite{Skrzypczyk2019_RobustnessMeasurement,Skrzypczyk2019_IncompatibleMeasurements,Takagi2019a}. 
On the other hand, the weight of resource is understood as the amount of resource needed to generate a given resource state by mixing up with some free state. 
This measure can also be computed with SDP for steering~\cite{Skrzypczyk2014}, coherence~\cite{Kaifeng2018}, and asymmetry~\cite{Kaifeng2018}. 
In addition, the weight of resource gives evidence of the usefulness of every quantum state in convex QRTs as well; Refs.~\cite{Ducuara2020,Uola2020} independently showed that the weight of resource quantifies the operational advantage of a given quantum state in a task called \textit{channel exclusion}. 
Thus, both resource measures enjoy computability and operational characterizations, indicating the usefulness of quantum resources. 

However, all of these previous results are based on the assumption of convexity, 
while physically well-motivated quantum resources do not necessarily have a convex structure.
For example, non-Gaussianity is an essential resource for quantum optics~\cite{Mattia2021}, including entanglement distillation~\cite{Eisert2002, Giedke2002, Fiuraifmmode2002, Zhang2010, Lami2018}, quantum error correction~\cite{Niset2009}, and universal quantum computation using continuous variable  systems~\cite{Lloyd1999,Bartlett2002,Ohliger2010,Menicucci2006,yamasaki2020polylogoverhead}, but the set of Gaussian states~\cite{Weedbrook2012,Adesso2014Gauss} is non-convex. 
(See Refs.~\cite{Takagi2018,Albarelli2018,PhysRevA.104.L020401} for another framework that aims to remedy the issue of non-convexity.)
Non-Markovianity is also regarded as a valuable resource for secure communication; a classical non-Markov chain facilitates the secret key agreement~\cite{Maurer1999}, 
and quantum non-Markovianity is used as a resource in the quantum one-time pad~\cite{Sharma2020}, but the set of quantum Markov states~\cite{Wakakuwa2017} is also non-convex.  
In addition, the sets of states with no quantum discord~\cite{Ollivier2001,Henderson2001}, e.g., classical-classical and classical-quantum states~\cite{ABC2016,Bera2017}, are not convex as well. 
More generally, it is not always possible to justify a mathematical assumption of convexity in tailored resource theories where physical limitations are identified from experimental constraints. While a probabilistic mixture of free states makes the set of free states convex~\cite{Takagi2018}, this approach cannot be generally accepted because classical randomness may itself be regarded as a resource in some cases~\cite{Groisman2005,Anshu2018}. 
Despite the existence of valuable resources beyond a convex structure, the operational characterization of the generalized robustness and the weight of resource, as well as the operational usefulness of quantum resources in more general QRTs without convexity assumption, have been largely unexplored. 
Indeed, the convexity of the set of resource-free states, which guarantees the connection between the generalized robustness and \textit{resource witness}~\cite{Brandao2005,Eisert_2007}, plays an essential role in the previous analyses. 
We cannot straightforwardly generalize these previous approaches when the convexity is not assumed. 
Thus, without new techniques and insights towards resources beyond convexity, the ultimate goal of QRTs, i.e., a unified understanding of fundamental advantages and limitations of quantum properties, could not be accomplished.

In a companion Letter~\cite{PRL} we show that the aforementioned challenges of generally non-convex QRTs can be overcome by providing a new characterization of the resourcefulness attributed to quantum states via the generalized robustness measure.
This paper complements it by providing additional results and explicit examples --- offering a detailed analysis and discussions of the results therein including full proofs of the main results ---  and further extends the framework to the regime not covered by Ref.~\cite{PRL}.
In particular, we extend the analysis to the weight of resource and provide its operational characterization in terms of multi-copy witnesses and channel exclusion tasks.
Furthermore, beyond QRTs for quantum states, here we develop extensions of our general results to \emph{dynamical} resources represented by quantum channels and instruments. 
We construct a multi-copy witness applying to a multi-copy of the Choi operator of a quantum channel in a similar way, based on the quantum channel version of the generalized robustness. Also, we show that the generalized robustness for quantum channels can be operationally interpreted as the worst-case advantage of some \textit{state discrimination} task. Our analysis can be further generalized to quantum instruments. 

Our results give a new perspective on the generalized robustness and the weight of resource in QRTs without assuming convexity. 
We establish a universal connection between these resource measures and resource witnesses that was only analyzed in restrictive convex cases previously.  
Thus, we open the direction of applying robustness- or weight-based techniques in QRTs to an even broader class of quantum resources that can be without the convexity assumption yet physically well-motivated. 
These results lead to a better understanding of general QRTs and shed new light on the framework that covers any resource regarded as physically well-motivated. 

The rest of this paper will be organized as follows. 
{\blue In Sec.~\ref{sec:prelim}, we introduce our notation, review background materials, and define the generalized robustness in general QRTs.}  
In Sec.~\ref{sec:multi-copy_witness}, we introduce a concept of multi-copy witness, and give a construction of multi-copy witnesses in a general QRT without convexity assumption, based on the generalized robustness. 
In Sec.~\ref{sec:advantage_multi-copy}, we show the operational advantage of every resource in a general QRT without convexity assumption QRT, for some channel discrimination task based on the multi-copy witness introduced in Sec.~\ref{sec:multi-copy_witness}. 
In Sec.~\ref{sec:worst-case}, we show that the generalized robustness in a general QRT without convexity assumption is operationally understood as the worst-case advantage of resource for channel discrimination. {\blue In Sec.~\ref{sec:special_case_qubit}, we explicitly show how our multi-copy witness looks in the single-qubit case,  give a geometric characterization of such witness, and evaluate the generalized robustness in a practically relevant instance of non-convex QRT\@.}
In Sec.~\ref{sec:extension_weight}, we present our analysis for the weight of resource; namely, we show another construction of multi-copy witness based on the weight of resource, an operational advantage of every quantum state in some channel exclusion task, an operational characterization of the weight of resource in QRTs without convexity assumption. 
In Sec.~\ref{sec:extension_channel}, we show an extension of our results to resource theories of quantum channels and quantum instruments. 
In Sec.~\ref{sec:conclusion}, we summarize our results and provide possible future directions.

\section{preliminaries}
\label{sec:prelim}
In this section, we introduce the notation of this paper and review background materials. 
\subsection{Notation}
In this paper, we use the letter $\mathcal{H}$ to denote a $d$-dimensional Hilbert space for finite $d$.
We use $\Herm(\mathcal{H})$ to denote the set of Hermitian operators on $\mathcal{H}$. 
We let $\mathcal{D}(\mathcal{H})$ denote the set of quantum states (density operators) on $\mathcal{H}$. 
For a linear operator $X$ on $\mathcal{H}$, we use $\|X\|_1$ to denote the trace norm of $X$ and $\|X\|_\infty$ the operator norm. 
The identity operator on a $d$-dimensional space is denoted by $I_{d}$. 

For Hilbert spaces $\mathcal{H}_1$ and $\mathcal{H}_2$, 
a completely positive and trace-preserving linear map from the set of linear operators on $\mathcal{H}_1$ to the set of those on $\mathcal{H}_2$ is called a \textit{quantum channel} from  $\mathcal{H}_1$ to $\mathcal{H}_2$. 
For example, we use $\mathcal{I}_{\mathcal{H}}$ to denote the identity operation on space $\mathcal{H}$. 
Let $\mathcal{O}(\mathcal{H}_1\to \mathcal{H}_2)$ denote the set of quantum channels from $\mathcal{H}_1$ to $\mathcal{H}_2$. 
Similarly, let $\mathcal{O}^{(\mathrm{CP})}(\mathcal{H}_1\to \mathcal{H}_2)$ denote the set of completely positive and trace-nonincreasing maps from $\mathcal{H}_1$ to $\mathcal{H}_2$. 
For a linear map $\Lambda$ from $\mathcal{H}_1$ to $\mathcal{H}_2$ with dimensions $d_1\coloneqq\dim(\mathcal{H}_1)$ and $d_2\coloneqq\dim(\mathcal{H}_2)$, the Choi operator $J_{\Lambda}$ is defined by 
\begin{equation}
    J_{\Lambda} \coloneqq \sum_{i,j = 1}^{d_1} \ketbra{i}{j}\otimes \Lambda(\ketbra{i}{j}). 
\end{equation}
For a subset $\mathcal{L}$ of linear operations, we define 
\begin{equation}
\label{eq:O_J}
    \mathcal{L}^{J} \coloneqq \{J_{\Lambda}: \Lambda \in \mathcal{L}\}; 
\end{equation}
for example, $\mathcal{O}^{J}(\mathcal{H}_1\to \mathcal{H}_2)$ denotes the set of Choi operators of quantum channels from $\mathcal{H}_1$ to $\mathcal{H}_2$.

\subsection{Quantum Resource Theory and Generalized Robustness}
We consider general quantum resource theories (QRTs) on a $d$-dimensional Hilbert space $\mathcal{H}$ for finite $d$. 
The set $\mathcal{F}(\mathcal{H})$ of free states is defined as some subset of $\mathcal{D}(\mathcal{H})$, and a state in $\mathcal{D}(\mathcal{H})\setminus\mathcal{F}(\mathcal{H})$ is called a resource state. 
For simplicity, we assume $\mathcal{F}(\mathcal{H})$ is closed. 
The generalized robustness is defined as follows. 
\begin{definition}
    Let $\rho \in \mathcal{D}(\mathcal{H})$ be a quantum state. Then, the generalized robustness $R_{\mathcal{F}(\mathcal{H})}(\rho)$ of $\rho$ with respect to the set $\mathcal{F}(\mathcal{H})$ of free states is defined by 
    \begin{equation}
        \label{eq:robustness}
        R_{\mathcal{F}(\mathcal{H})}(\rho) = \min_{\tau \in \mathcal{D}(\mathcal{H})}\left\{s\geq 0 : \dfrac{\rho + s \tau}{1 + s} \eqqcolon \sigma \in \mathcal{F}(\mathcal{H}) \right\}. 
    \end{equation}
\end{definition}
By definition, $R_{\mathcal{F}(\mathcal{H})}(\rho)$ is non-negative. 
Moreover, the generalized robustness is \textit{faithful}, \textit{i.e.}, $R_{\mathcal{F}(\mathcal{H})}(\rho) = 0$ if and only if $\rho \in \mathcal{F}(\mathcal{H})$. 
{\blue 
In this paper, we assume that $\mathcal{F}(\mathcal{H})$ contains at least one full-rank state; otherwise, $R_{\mathcal{F}(\mathcal{H})}(\rho) = \infty$ for any resource state $\rho \in \mathcal{D}(\mathcal{H})\backslash \mathcal{F}(\mathcal{H})$.}
Note that we do not assume convexity of $\mathcal{F}(\mathcal{H})$, at variance with previous literature in which the generalized robustness has been studied for convex resource theories. 

{\blue
When $\mathcal{F}(\mathcal{H})$ is convex, by considering the convex dual~\cite{boyd_vandenberghe_2004} of~\eqref{eq:robustness}, we have another expression of the generalized robustness~\cite{Brandao2005,Regula2017, Napoli2016,Piani2016,Takagi2019a,Takagi2019b}, in particular, for any $\rho \in \mathcal{D}(\mathcal{H})\backslash\mathcal{F}(\mathcal{H})$,  
\begin{equation}
    \label{eq:robustness_witness}
    \begin{aligned}
         R_{\mathcal{F}(\mathcal{H})}(\rho) =&  \max\Big\{\tr\left[W\rho\right] - 1\\
         & : W\geq 0,\,\, \tr\left[W\sigma\right] \leq 1, \forall \sigma \in \mathcal{F}(\mathcal{H})\Big\}. 
    \end{aligned}
\end{equation}
The above expression relates the evaluation of the generalized robustness to the expectation value of an optimal {\em witness} operator $W$.

Without the assumption of convexity, we no longer have this equivalence relation in general, but we can still show that the right-hand side of~\eqref{eq:robustness_witness} lower-bounds the generalized robustness. Specifically, for a general QRT with an arbitrary set $\mathcal{F}(\mathcal{H})$ of free states, we have for any $\rho \in \mathcal{D}(\mathcal{H})\backslash\mathcal{F}(\mathcal{H})$,  
    \begin{equation}
        \label{eq:robustness_lower_bound_main}
        \begin{aligned}
        R_{\mathcal{F}(\mathcal{H})}(\rho) \geq & \max\Big\{\tr\left[W\rho\right] - 1 \\
        & : W\geq 0,\,\, \tr\left[W\sigma\right] \leq 1, \forall \sigma \in \mathcal{F}(\mathcal{H})\Big\}. 
        \end{aligned}
    \end{equation}
 This is a consequence of a well-known property in convex analysis called weak duality~\cite{boyd_vandenberghe_2004}---a self-contained detailed derivation of this result is given in Appendix~\ref{appendix:generalized_robustness_lower_bound}.  We remark that Eq.~(\ref{eq:robustness_lower_bound_main}) entails that for an arbitrary resource theory, even in cases in which the exact evaluation of the generalized robustness is mathematically or experimentally impractical --- as it may happen, e.g., for high dimensional systems --- we can always obtain accessible lower bounds to $R_{\mathcal{F}(\mathcal{H})}(\rho)$, by means of witnesses operators acting on a single copy of the state $\rho$. 
Note that the lower bound obtained in~\eqref{eq:robustness_lower_bound_main} is typically loose when $\mathcal{F}(\mathcal{H})$ is non-convex. 
Indeed, by definition, the right-hand side of~\eqref{eq:robustness_lower_bound_main} is zero, by taking $W=I_d$, for states in the convex hull of $\mathcal{F}(\mathcal{H})$. 
Thus, when $\mathcal{F}(\mathcal{H})$ is non-convex, the lower bound fails to detect some resource states, while the generalized robustness $R_{\mathcal{F}(\mathcal{H})}$ is a faithful measure which becomes zero only for free states in any QRT.  

In the following, we will enlarge our analysis to define a class of witness operators acting on multiple copies of quantum states, which are intrinsically constructed in terms of general resource measures such as the generalized robustness and the weight of resource.
}

\subsection{Generalized Bloch Vector}
In this section, we review the Bloch-vector characterization of quantum states on a $d$-dimensional Hilbert space. 
Consider a general $d$-dimensional Hilbert space $\mathcal{H} \simeq \mathbb{C}^{d}$ for $d\geq 2$. 
Given a quantum state $\eta$, we consider a generalized Bloch vector $(x_1,x_2,\ldots,x_{d^2-1})$ defined by 
\begin{equation}
    \eta = \frac{1}{d}\left(I_d + \sqrt{\dfrac{d(d-1)}{2}} \sum_{j=1}^{d^2-1} x_j\lambda_j \right). 
\end{equation}
Here, $I_d$ is the $d\times d$ identity matrix, and $\{\lambda_j\}_j$ are the $d\times d$ generalized Gell-Mann matrices~\cite{KIMURA2003339,Bertlmann_2008} satisfying 
\begin{align}
    &\lambda_j = \lambda_j^\dagger, &&j = 1,2,\ldots,d^2-1, \\ 
    &\tr[\lambda_j] = 0, &&j = 1,2,\ldots,d^2-1, \\ 
    &\tr[\lambda_i^\dagger \lambda_j] = 2\delta_{ij}, &&i,j = 1,2,\ldots,d^2-1, 
\end{align}
where $\dagger$ denotes the Hermitian conjugate of a matrix, and $\delta_{ij}$ is the Kronecker delta.  

The following lemma by Byrd and Khaneja~\cite{Byrd2003_characterization} gives a necessary and sufficient condition for a valid generalized Bloch vector $(x_j)_j$. 
\begin{lemma}[\cite{Byrd2003_characterization}]~\label{lem:bloch_ball}
    Let $\mathcal{H}$ be a $d$-dimensional Hilbert space.
    Let $\eta$ be a Hermitian operator on $\mathcal{H}$ written in the form 
    \begin{equation}
        \eta = \frac{1}{d}\left(I_d + \sqrt{\dfrac{d(d-1)}{2}} \sum_{j=1}^{d^2-1} x_j\lambda_j \right).
    \end{equation}
    Then, $\eta$ is positive semidefinite if and only if 
    \begin{equation}
        S_m(\eta) \geq 0
    \end{equation}
for all $m = 1,2,\ldots,d$, where $S_m(\eta)$ is defined in the following recursive way: 
\begin{align}
    \label{eq:s_k}
    &S_m(\eta) \coloneqq \dfrac{1}{m}\sum_{l=1}^{m} \left((-1)^{l-1}\tr\left[\eta^l\right]S_{m-l}(\eta)\right) 
\end{align}
for $m\geq 1$ with 
\begin{equation}
    S_0(\eta) \coloneqq 1. 
\end{equation}
\end{lemma}

\section{Multi-copy witness in QRTs without convexity assumptions}
\label{sec:multi-copy_witness}
In this section, we construct our multi-copy witnesses in general QRTs without convexity restriction. 
The concept of multi-copy witness was originally introduced in entanglement theory~\cite{Horodecki2003, Horodecki2009}.
In this paper, we show a construction of a $m$-copy witness for each $m = 2,3,\ldots,d$. 
For an arbitrary free state $\sigma$, there exists at least one $m$ such that the $m$-copy witness can discern a given state $\rho$ from $\sigma$.  

\begin{theorem}[Theorem~2 in the companion Letter~\cite{PRL}]
\label{cor:high_order_witness_qudit_conventional}
    Let $\mathcal{H}$ be a $d$-dimensional Hilbert space. 
    Let $\rho \in \mathcal{D}(\mathcal{H})\backslash\mathcal{F}(\mathcal{H})$ be a resource state.  
    Then, for $s>0$, we can construct a family $\widetilde{\mathcal{W}}_{\rho,s} \coloneqq (\widetilde{W}_{m}(\rho,s) \in \Herm(\mathcal{H}^{\otimes m}): m = 2,3,\ldots,d)$ of Hermitian operators with the following property: 
    \begin{align}
    \label{eq:conventional_witness_resource}
    &\max_{m = 2,3,\ldots,d}\tr\left[\widetilde{W}_{m}(\rho,s)\rho^{\otimes m}\right] < 0, \\ 
    \label{eq:conventional_witness_free}
    &\max_{m = 2,3,\ldots,d} \tr\left[\widetilde{W}_{m}(\rho,s)\sigma^{\otimes m}\right] \geq 0, \,\,\forall \sigma \in \mathcal{F}(\mathcal{H})
\end{align}
    if and only if
    \begin{equation}
        s < R_{\mathcal{F}(\mathcal{H})}(\rho).
    \end{equation}
\end{theorem}

To show this theorem, we will prove the following slightly different proposition for simplicity of the argument, so that the theorem should be derived from this proposition.
\begin{proposition}
\label{thm:high_order_witness_qudit}
    Let $\mathcal{H}$ be a $d$-dimensional Hilbert space. 
    Let $\rho \in \mathcal{D}(\mathcal{H})\backslash\mathcal{F}(\mathcal{H})$ be a resource state.  
    Then, for $s > 0$, we can construct a family $\mathcal{W}_{\rho,s} \coloneqq (W_m(\rho,s) \in \Herm(\mathcal{H}^{\otimes m}): m = 2,3,\ldots,d)$ of Hermitian operators on $\mathcal{H}^{\otimes m}$ with the following properties. 
    \begin{enumerate}
    \item For all $m = 2,3,\ldots,d$ and all $s>0$, 
    \begin{equation}
    \label{eq:high_order_witness_resource_qudit}
    \tr\left[W_m(\rho,s)\rho^{\otimes m}\right] \geq 0. 
    \end{equation}
    
    \item For each $\sigma \in \mathcal{F}(\mathcal{H})$, there exists $m$ such that 
    \begin{equation}
    \label{eq:high_order_witness_free_qudit}
    \tr\left[W_m(\rho,s)\sigma^{\otimes m}\right] < 0 
    \end{equation}
    if and only if
    \begin{equation}
        s < R_{\mathcal{F}(\mathcal{H})}(\rho).
    \end{equation}
    \end{enumerate}
\end{proposition}

Before proving Proposition~\ref{thm:high_order_witness_qudit}, we describe a way to construct a witness $\widetilde{W}_{m}(\rho,s)$ in Theorem~\ref{cor:high_order_witness_qudit_conventional} from a Hermitian operator $W_{m}(\rho,s)$ in Proposition~\ref{thm:high_order_witness_qudit}. 

\begin{proof}[Proof of Theorem~\ref{cor:high_order_witness_qudit_conventional}]
Under the assumption that we have $W_m(\rho,s)$ in Proposition~\ref{thm:high_order_witness_qudit}, we define 
\begin{equation}
    \widetilde{W}_{m}(\rho,s) \coloneqq -C\left(W(\rho,s) + \Delta_{m}(\rho,s) I_d^{\otimes m}\right), 
\end{equation}
with 
\begin{equation}
    \Delta_{m}(\rho,s) \coloneqq \min_{\sigma \in \mathcal{F}(\mathcal{H})} 
    \left|\tr[W_m(\rho,s)\sigma^{\otimes m}]\right| (> 0), 
\end{equation}
where $C>0$ is any suitable constant for normalization. 
With this definition, 
\begin{equation}
    \begin{aligned}
        \tr[\widetilde{W}_m(\rho,s)\rho^{\otimes m}] 
        &= -C\tr[W_m(\rho,s)\rho^{\otimes d}] -C\Delta_{m}(\rho,s)  \\ 
        &< -C\tr[W_m(\rho,s)\rho^{\otimes d}] \\ 
        &\leq 0. 
    \end{aligned}
\end{equation}
Moreover, by the definition of $\Delta_{m}(\rho,s)$, for free states $\sigma$ satisfying $\tr[W_m(\rho,s)\sigma^{\otimes m}] < 0$, 
\begin{equation}
    \begin{aligned}
        \tr[\widetilde{W}_m(\rho,s)\sigma^{\otimes d}] = C\left(-\tr[W_m(\rho,s)\sigma^{\otimes d}] - \Delta_{m}(\rho,s) \right) \geq 0.  
    \end{aligned}
\end{equation}
For free states $\sigma$ with $\tr[W_m(\rho,s)\sigma^{\otimes m}] \geq 0$, 
we have 
\begin{equation}
    \begin{aligned}
        \tr[\widetilde{W}_m(\rho,s)\sigma^{\otimes d}] = C\left(-\tr[W_m(\rho,s)\sigma^{\otimes d}] - \Delta_{m}(\rho,s) \right) < 0.  
    \end{aligned}
\end{equation}

Hence, from Proposition~\ref{thm:high_order_witness_qudit}, for $m = 2,3,\ldots,d$, there exists a family $\widetilde{\mathcal{W}}_{\rho,s} \coloneqq (\widetilde{W}_m(\rho,s): m = 2,3,\ldots,d)$ of Hermitian operators with the following properties. 
    \begin{enumerate}
    \item For all $m = 2,3,\ldots,d$ and all $s>0$, 
    \begin{equation}
    \tr\left[\widetilde{W}_m(\rho,s)\rho^{\otimes m}\right] < 0. 
    \end{equation}
    
    \item For each $\sigma \in \mathcal{F}(\mathcal{H})$, there exists $m$ such that 
    \begin{equation}
    \tr\left[\widetilde{W}_m(\rho,s)\sigma^{\otimes m}\right] \geq 0 
    \end{equation}
    if and only if
    \begin{equation}
        s < R_{\mathcal{F}(\mathcal{H})}(\rho).
    \end{equation}
    \end{enumerate}
Here, for a state $\eta$, 
\begin{equation}
    \tr\left[\widetilde{W}_m(\rho,s)\eta^{\otimes m}\right] < 0,\,\, \forall m = 2,3,\ldots, m
 \end{equation}
is equivalent to 
 \begin{equation}
    \max_{m = 2,3,\ldots,d}\tr\left[\widetilde{W}_{m}(\rho,s)\eta^{\otimes m}\right] < 0, 
\end{equation}
and the existence of $m$ such that
\begin{equation}
    \tr\left[\widetilde{W}_m(\rho,s)\eta^{\otimes m}\right] \geq 0 
\end{equation}
is equivalent to 
\begin{equation}
    \max_{m = 2,3,\ldots,d}\tr\left[\widetilde{W}_{m}(\rho,s)\eta^{\otimes m}\right] \geq 0.  
\end{equation}
Thus, we have Theorem~\ref{cor:high_order_witness_qudit_conventional} from Proposition~\ref{thm:high_order_witness_qudit}. 
\end{proof}

In the proof of Proposition~\ref{thm:high_order_witness_qudit}, we use the following lemma, which we prove in Appendix~\ref{app:proof_witness_characterization}. 

\begin{lemma}~\label{lem:witness_characterization}
    Let $\rho$ be a quantum state, and let $s > 0$. 
    Consider a quantum state $\eta$.  
    Then, 
    $\eta$ can be expressed as 
    \begin{equation}
        \eta = \frac{\rho + s' \tau}{1 + s'}
    \end{equation}
    by using some positive number $0 < s' \leq s$ and some state $\tau \in \mathcal{D}(\mathcal{H})$ 
    if and only if it holds for all $m = 1,2,\ldots,d$ that
    \begin{equation}
        S_m\left(\frac{1+s}{s} \eta - \frac{1}{s}\rho \right) \geq 0,
    \end{equation}
     with $S_m$ given in Lemma~\ref{lem:bloch_ball}. 
\end{lemma}

Now, we give the proof of Proposition~\ref{thm:high_order_witness_qudit}. 
Intuitively, our construction of $W_m$ for Proposition~\ref{thm:high_order_witness_qudit} is intended to witness a valid density operator in the state space, so that we can use the witnessed quantum state as a feasible solution in the minimization in the definition of generalized robustness in Eq.~\eqref{eq:robustness}. 

\begin{proof}[Proof of Proposition~\ref{thm:high_order_witness_qudit}]

For a state $\eta \in \mathcal{D}(\mathcal{H})$, define 
\begin{equation}
    S_{m,\rho,s}(\eta) \coloneqq S_m\left(\frac{1+s}{s} \eta - \frac{1}{s}\rho \right) 
\end{equation}
for $m = 1,2,\ldots,d$.
Note that it holds automatically that
\begin{equation}
    S_{1,\rho,s}(\eta) = \frac{1+s}{s}\tr[\eta] - \frac{1}{s}\tr[\rho] = 1\geq 0, 
\end{equation}
and thus we will discuss whether $S_{m,\rho,s}(\eta)\geq 0$ for $m\geq 2$.

It trivially follows that $S_{m,\rho,s}(\rho) = S_m(\rho) \geq 0$ for all $m$ since $\rho$ is a state, in particular, a positive semidefinite operator. 
When $s < R_{\mathcal{F}(\mathcal{H})}(\rho)$, by the definition of the generalized robustness, for any free state $\sigma \in \mathcal{F}(\mathcal{H})$, there is no positive number $s' \leq  s$ such that 
\begin{equation}
    \sigma = \frac{\rho + s'\tau}{1+s'}
\end{equation}
with some $\tau \in \mathcal{D}(\mathcal{H})$.
By Lemma~\ref{lem:witness_characterization}, in this case, there exists $2 \leq m \leq d$ such that $S_{m,\rho,s}(\sigma) < 0$.
On the other hand, when $s \geq R_{\mathcal{F}(\mathcal{H})}(\rho)$, there exists $0 \leq s' \leq s$ such that 
\begin{equation}
    \sigma = \frac{\rho + s'\tau}{1+s'}
\end{equation}
with some $\tau \in \mathcal{D}(\mathcal{H})$.
Indeed, we can take a free state $\sigma$ achieving $R_{\mathcal{F}(\mathcal{H})}(\rho)$ with $s' = R_{\mathcal{F}(\mathcal{H})}(\rho)$. 
In this case, for all $m = 2,3,\ldots, d$, 
we have $S_{m,\rho,s}(\sigma)  \geq 0$ for such $\sigma$. 

Thus, it suffices to construct $W_m(\rho,s)$ satisfying 
\begin{equation}
    \tr\left[W_{m}(\rho,s) \eta^{\otimes n}\right] = S_{m,\rho,s}(\eta)
\end{equation}
for $\eta \in \mathcal{D}(\mathcal{H})$, 
for all $2\leq m\leq d$ and $s>0$. 
Using the generalized Bloch vector $(x_j)_j$ of $\eta$, $S_{m,\rho,s}(\eta)$ is expressed as 
\begin{equation}
    \begin{aligned}
        &S_{m,\rho,s}(\eta) = \sum_{l = 0}^m \left(\sum_{\substack{n_1,\ldots,n_{d^2-1} \geq 0 \\ n_1 + n_2 \cdots n_{d^2-1} = l}} \left(c^{(l,m)}_{n_1,n_2,\ldots,n_{d^2-1}} \prod_{j=1}^{d^2-1} \left(x_j\right)^{n_j}\right)\right)
    \end{aligned}
\end{equation}
with some real coefficients $c^{(l,m)}_{n_1,n_2,\ldots,n_{d^2-1}}$. 
By definition of $S_m$, $S_{m,\rho,s}(\eta)$ is a real $m$-degree polynomial of $(x_j)_j$. 
Now, define a Hermitian operator 
    \begin{equation}
    \begin{aligned}
        \label{eq:construction_witness_fraction}
        &W^{(l,m)}_{n_1,n_2,\ldots,n_{d^2-1}} \\ 
        &\coloneqq \left(\dfrac{d}{2(d-1)}\right)^{\tfrac{l}{2}}c^{(l,m)}_{n_1,n_2,\ldots,n_{d^2-1}} \\ 
        &\quad\cdot \left(\lambda_1\right)^{\otimes n_1} \otimes\left(\lambda_2\right)^{\otimes n_2} \cdots \otimes \left(\lambda_{d^2-1}\right)^{\otimes n_{d^2-1}}\otimes I^{\otimes (m-l)}
    \end{aligned}
\end{equation}
on $\mathcal{H}^{\otimes m}$. 
In Appendix~\ref{app:witness_derivation}, we show that this operator $W^{(l,m)}_{n_1,n_2,\ldots,n_{d^2-1}}$ satisfies 
\begin{equation}
 \label{eq:witness_derivation}
    \tr\left[W^{(l,m)}_{n_1,n_2,\ldots,n_{d^2-1}}\eta^{\otimes m}\right] = \left(c^{(l,m)}_{n_1,n_2,\ldots,n_{d^2-1}} \prod_{j=1}^{d^2-1} \left(x_j\right)^{n_j}\right).
\end{equation}
Thus, one candidate of $W(\rho,s)$ can be constructed as 
\begin{equation}
\label{eq:W_proof}
    W_{m}(\rho,s) = \sum_{l = 0}^m \left(\sum_{\substack{n_1,\ldots,n_{d^2-1} \geq 0 \\ n_1 + n_2 \cdots n_{d^2-1} = l}} W^{(l,m)}_{n_1,n_2,\ldots,n_{d^2-1}} \right), 
\end{equation}
and this operator satisfies 
\begin{equation}
\tr\left[W_{m}(\rho,s) \eta^{\otimes n}\right] = S_{m,\rho,s}(\eta),
\end{equation}
which yields the conclusion.
\end{proof}

We remark that if one wants to make $W_{m}(\rho,s)$ more symmetric than Eq.~\eqref{eq:W_proof}, 
one may give another construction of $W_{m}(\rho,s)$ as 
\begin{equation}
    \label{eq:d_deg_witness_qudit}
    W_{m}(\rho,s) = \sum_{l = 0}^m \left(\sum_{\substack{n_1,\ldots,n_{d^2-1} \geq 0 \\ n_1 + n_2 \cdots n_{d^2-1} = l}} \tilde{W}^{(l,m)}_{n_1,n_2,\ldots,n_{d^2-1}} \right)
\end{equation}
with 
    \begin{equation}
    \begin{aligned}
        &\tilde{W}^{(l,m)}_{n_1,n_2,\ldots,n_{d^2-1}} \\
        &\coloneqq \left(\dfrac{d}{2(d-1)}\right)^{\tfrac{l}{2}}c^{(l,m)}_{n_1,n_2,\ldots,n_{d^2-1}} \\ 
        &\quad\cdot\dfrac{n_1!n_2!\cdots n_{d^2-1}!(m-l)!}{m!}\sum\lambda^{(l,m)}_{n_1,n_2,\ldots,n_{d^2-1}}, 
    \end{aligned}
\end{equation}
where the sum is taken over all trace-less Hermitian operators $\lambda^{(l,m)}_{n_1,n_2,\ldots,n_{d^2-1}}$ obtained by permuting the systems of $\left(\lambda_1\right)^{\otimes n_1} \otimes\left(\lambda_2\right)^{\otimes n_2} \cdots \otimes \left(\lambda_{d^2-1}\right)^{\otimes n_{d^2-1}} \otimes I^{\otimes (m-l)}$. 
This symmetric construction of $W_{m}(\rho,s)$ can also be used for Proposition~\ref{thm:high_order_witness_qudit}.

From Theorem~\ref{cor:high_order_witness_qudit_conventional}, we have 
\begin{equation}~\label{eq:witness_robustness_relation}
    \begin{aligned}
        \sup\left\{s>0: \widetilde{\mathcal{W}}_{\rho,s}\,\,\mathrm{is\,\,a\,\,witness} \right\} 
        = R_{\mathcal{F}(\mathcal{H})}(\rho). 
    \end{aligned}
\end{equation}
We can also see that when $s$ is closer to $R_{\mathcal{F}(\mathcal{H})}$, $W_{m}(\rho,s)$ serves as a \textit{better} witness; that is, $W_{m}(\rho,s)$ recognizes more resource states, as shown in the following proposition. 
\begin{proposition}
    If $0 < s' < s$, then 
    \begin{equation}
        \begin{aligned}
        \label{eq:good_witness_qudit}
        &\{\eta \in \mathcal{D}(\mathcal{H}): \tr[W_m(\rho,s') \eta^{\otimes m}] \geq 0,\,\, \forall m=2,3,\ldots,d \} \\
        &\subsetneq \{\eta \in \mathcal{D}(\mathcal{H}): \tr[W_m(\rho,s) \eta^{\otimes m}] \geq 0,\,\,\forall m=2,3,\ldots,d \}. 
        \end{aligned}
    \end{equation}
\end{proposition}

\begin{proof}
    We first show the inclusion, \textit{i.e.}, 
    \begin{equation}
        \begin{aligned}
        &\{\eta \in \mathcal{D}(\mathcal{H}): \tr[W_m(\rho,s') \eta^{\otimes m}] \geq 0,\,\, \forall m=2,3,\ldots,d \} \\
        &\subseteq \{\eta \in \mathcal{D}(\mathcal{H}): \tr[W_m(\rho,s) \eta^{\otimes m}] \geq 0,\,\,\forall m=2,3,\ldots,d \}. 
        \end{aligned}
    \end{equation}
    Suppose that 
    \begin{equation}
        \eta \in \{\eta \in \mathcal{D}(\mathcal{H}): \tr[W_m(\rho,s') \eta^{\otimes m}] \geq 0,\,\, \forall m=2,3,\ldots,d \}. 
    \end{equation}
    Then, by the construction of $W_{m}(\rho,s)$, for all $m = 2,3,\ldots,d$, 
    \begin{equation}
        S_{m,\rho,s'}(\eta) = S_m\left(\frac{1 + s'}{s'}\eta - \frac{1}{s'}\rho\right) \geq 0. 
    \end{equation}
    By Lemma~\ref{lem:witness_characterization}, there exists $0 < s''\leq s'$ and a state $\tau \in \mathcal{D}(\mathcal{H})$ such that 
    \begin{equation}
        \eta = \frac{\rho + s''\tau}{1 + s''}. 
    \end{equation}
    Since we also have $s'' < s$, again from Lemma~\ref{lem:witness_characterization}, 
    for all $m = 2,3,\ldots,d$, 
    \begin{equation}
        S_{m,\rho,s}(\eta) = S_m\left(\frac{1 + s}{s}\eta - \frac{1}{s}\rho\right) \geq 0; 
    \end{equation}
    that is, 
     \begin{equation}
        \eta \in \{\eta \in \mathcal{D}(\mathcal{H}): \tr[W_m(\rho,s) \eta^{\otimes m}] \geq 0,\,\, \forall m=2,3,\ldots,d \}, 
    \end{equation}
    which shows the inclusion. 

    Now, we show that the inclusion is strict.
    In particular, we construct a state (\textit{i.e.,} $\eta_s$ defined below) included in
    \begin{equation}\{\eta \in \mathcal{D}(\mathcal{H}): \tr[W_m(\rho,s) \eta^{\otimes m}] \geq 0,\,\,\forall m=2,3,\ldots,d \}
    \end{equation}
    but not in
    \begin{equation}
        \{\eta \in \mathcal{D}(\mathcal{H}): \tr[W_m(\rho,s') \eta^{\otimes m}] \geq 0,\,\, \forall m=2,3,\ldots,d \}.
    \end{equation}
    For the construction, choose any pure state $\phi \in \mathcal{D}(\mathcal{H})$ satisfying $\phi \neq \rho$.
    We construct a state $\eta_s$ by 
    \begin{equation}
        \eta_{s} \coloneqq \frac{\rho + s\phi}{1 + s}. 
    \end{equation}
    By Lemma~\ref{lem:witness_characterization}, using a similar argument as above, 
    \begin{equation}
    \label{eq:eta_s}
        \eta_s \in \{\eta \in \mathcal{D}(\mathcal{H}): \tr[W_m(\rho,s) \eta^{\otimes m}] \geq 0,\,\, \forall m=2,3,\ldots,d \}.  
    \end{equation}
    On the other hand, we show that $\eta_s$ cannot be written as 
    \begin{equation}
        \eta_{s} = \frac{\rho + s'\tau}{1 + s'}
    \end{equation}
    for any $s'$ satisfying $0 < s' < s$ and any $\tau \in \mathcal{D}(\mathcal{H})$. 
    By way of contradiction, suppose that this is possible. 
    Then, we have 
    \begin{equation}
        \frac{\rho + s\phi}{1 + s} = \frac{\rho + s'\tau}{1 + s'}; 
    \end{equation}
    that is, 
    \begin{equation}
        \left(\frac{1}{1 + s'} - \frac{1}{1 + s} \right)\rho = \frac{s}{1 + s}\phi - \frac{s'}{1 + s'}\tau. 
    \end{equation}
    Since $s' < s$, 
    \begin{equation}
        \left(\frac{1}{1 + s'} - \frac{1}{1 + s} \right)\rho
    \end{equation}
    is positive, so 
    \begin{equation}
        \frac{s}{1 + s}\phi - \frac{s'}{1 + s'}\tau
    \end{equation}
    must be positive as well. 
    At the same time, since $\phi$ is chosen to be a pure state, for this operator to be positive, it is necessary that $\tau = \phi$. 
    This choice leads to $\rho = \phi$, but this contradicts $\phi \neq \rho$. 
    Therefore, there is no $s' < s$ with which we can write 
   \begin{equation}
        \eta_{s} = \frac{\rho + s'\tau}{1 + s'}. 
    \end{equation}
    Hence, by Lemma~\ref{lem:witness_characterization}, there exists $2\leq m \leq d$ such that 
    \begin{equation}
        S_{m,\rho,s'}(\eta_s) = S_m\left(\frac{1 + s'}{s'}\eta_s - \frac{1}{s'}\rho\right) < 0. 
    \end{equation}
    Thus, we have 
    \begin{equation}
    \label{eq:eta_s_prime}
        \eta_{s} 
            \not\in\{\eta \in \mathcal{D}(\mathcal{H}): \tr[W_m(\rho,s') \eta^{\otimes m}] \geq 0,\,\, \forall m=2,3,\ldots,d \}.
    \end{equation}
    Therefore, Eqs.~\eqref{eq:eta_s} and~\eqref{eq:eta_s_prime} show that the inclusion is strict. 
\end{proof}

\section{Operational advantage of resources in multi-copy channel discrimination}
\label{sec:advantage_multi-copy}
In this section, we show that the multi-copy witnesses obtained in the previous sections lead to an operational advantage of all resource states in general QRTs without convexity assumption for a variant of channel discrimination tasks, which we call \textit{$m$-input channel discrimination}. 

In $m$-input channel discrimination, one aims to distinguish channels that act on $m$ copies $\rho^{\otimes m}$ of a given state $\rho$. 
Let $\{p_i,\Lambda^{(m)}_i\}_i$ be an ensemble of channels, where channels $\Lambda^{(m)}_i$ on $\mathcal{D}(\mathcal{H}^{\otimes m})$ are randomly picked with probability $p_i$. 
Once a channel $\Lambda^{(m)}_i$ is sampled, we apply this channel to the $m$ copies of the given state $\rho$; that is, at this point, we have $\Lambda^{(m)}_i(\rho^{\otimes m})$. 
Our goal is to figure out which channel from the channel ensemble acted on the state. 
To identify the label $i$, we perform a quantum measurement $\{M_i\}_i$ to the resulting state. 
The success probability for this $m$-input channel discrimination task is given by 
\begin{equation}
    \label{eq:success_prob}
    p_{\mathrm{succ}}(\{p_i,\Lambda^{(m)}_i\}_i, \{M_i\}_i, \rho^{\otimes m}) \coloneqq \sum_{i} p_i \tr\left[M_i~\Lambda^{(m)}_i(\rho^{\otimes m})\right]. 
\end{equation}
In this task, given a quantum state and a channel ensemble, we aim to maximize the success probability by choosing the best measurement strategy. 
We characterize the operational advantage of a resource state in this scenario by taking the ratio between the best success probability with respect to the given state and 
the best success probability with respect to a free state. 

The separation of resource state $\rho$ and free states $\sigma$ shown in Proposition~\ref{thm:high_order_witness_qudit}  implies all resource states in general QRTs without convexity restriction show operational advantages in this task, as shown in the following theorem. 

\begin{theorem}[Theorem~3 in the companion Letter~\cite{PRL}]
\label{thm:advantage_multi-copy}
Let $\mathcal{H}$ be a $d$-dimensional Hilbert space. 
For any resource state $\rho \in \mathcal{D}(\mathcal{H})\backslash\mathcal{F}(\mathcal{H})$, there exists a family of channel ensembles $\left(\{p_i,\Lambda^{(m)}_i\}_i\right)_{m=2}^{d}$ such that 
\begin{equation}
    \min_{\sigma \in \mathcal{F}(\mathcal{H})}\max_{m = 2,3,\ldots,d}\dfrac{\displaystyle \max_{\{M^{(m)}_i\}_i} p_{\mathrm{succ}}(\{p_i,\Lambda^{(m)}_i\}_i,\{M^{(m)}_i\}_i,\rho^{\otimes m})}{\displaystyle \max_{\{M^{(m)}_i\}_i} p_{\mathrm{succ}}(\{p_i,\Lambda^{(m)}_i\}_i,\{M^{(m)}_i\}_i,\sigma^{\otimes m})} > 1. 
\end{equation}
\end{theorem}

\begin{proof}
Let $\{W_m\}_{m=2}^d$ be a family of Hermitian operators such that every free state $\sigma\in\mathcal{F}(\mathcal{H})$ comes with an integer $m\in\{2,\dots, d\}$ satisfying 
\begin{align}
        &\tr[W_m\rho^{\otimes m}] > 1, \\ 
        &0\leq\tr[W_m\sigma^{\otimes m}] \leq 1.  
        \label{eq:modified witness free}
    \end{align}
    The existence of such a family of Hermitian operators is guaranteed by Theorem~\ref{cor:high_order_witness_qudit_conventional}.
Indeed, it ensures that when $s < R_{\mathcal{F}(\mathcal{H})}(\rho)$, for every $\sigma\in\mathcal{F}(\mathcal{H})$ there exists $m\in\{2,\ldots, d\}$ satisfying
    \begin{align}
        \label{eq:witness_resource}
        &\tr[\widetilde{W}_m(\rho,s)\rho^{\otimes m}] < 0, \\ 
        &\tr[\widetilde{W}_m(\rho,s)\sigma^{\otimes m}] \geq 0 
    \end{align}
    for the family $\widetilde{\mathcal{W}}_{\rho,s} = (\widetilde{W}_{m}(\rho,s) \in \Herm(\mathcal{H}^{\otimes m}): m = 2,3,\ldots,d)$ of Hermitian operators.
    Then, we may take $W_m \coloneqq I_d^{\otimes m} - \widetilde{W}_m(\rho,s)/\|\widetilde{W}_m(\rho,s)\|_\infty$ for some $s < R_{\mathcal{F}(\mathcal{H})}(\rho)$. 
    Note that Eq.~\eqref{eq:witness_resource} implies that $\widetilde{W}_m(\rho,s)$ has at least one negative eigenvalues, so we have $\|W_m\|_\infty > 1$. 
    Define two quantum channels $\Lambda^{(m)}_1$ and $\Lambda^{(m)}_2$ as 
        \begin{align}
            &\Lambda^{(m)}_1(X)\\ \nonumber
            &= \left(\dfrac{\tr(X)}{2} + \dfrac{\tr(W_mX)}{2\|W_m\|_\infty}\right)\ketbra{0}{0} + \left(\dfrac{\tr(X)}{2} - \dfrac{\tr(W_mX)}{2\|W_m\|_\infty}\right)\ketbra{1}{1}, \\
            &\Lambda^{(m)}_2(X)\\ \nonumber 
            &= \left(\dfrac{\tr(X)}{2} - \dfrac{\tr(W_mX)}{2\|W_m\|_\infty}\right)\ketbra{0}{0} + \left(\dfrac{\tr(X)}{2} + \dfrac{\tr(W_mX)}{2\|W_m\|_\infty}\right)\ketbra{1}{1}. 
        \end{align}
    Let $\sigma\in\mathcal{F}(\mathcal{H})$ be an arbitrary free state and let $m$ be an integer that satisfies \eqref{eq:modified witness free}. 
    Then, since $\tr[W_m\rho^{\otimes m}] >1$, 
    \begin{equation}
        \begin{aligned}
            &\|\Lambda^{(m)}_1(\rho^{\otimes m}) - \Lambda^{(m)}_2(\rho^{\otimes m})\|_1 \\ 
            &= \left\|\dfrac{\tr[W_m\rho^{\otimes d}]}{\|W_m\|_\infty}\ketbra{0}{0} - \dfrac{\tr[W_m\rho^{\otimes m}]}{\|W_m\|_\infty}\ketbra{1}{1}\right\|_1 \\ 
            &= \dfrac{2\tr[W_m\rho^{\otimes m}]}{\|W_m\|_\infty} \\ 
            &> \dfrac{2}{\|W_m\|_\infty}. 
        \end{aligned}
    \end{equation}
    Similarly, since $0\leq\tr[W_m\sigma^{\otimes m}] \leq 1$, 
    \begin{equation}
        \begin{aligned}
        &\|\Lambda^{(m)}_1(\sigma^{\otimes m}) - \Lambda^{(m)}_2(\sigma^{\otimes m})\|_1 \\ 
            &= \dfrac{2\tr[W_m\sigma^{\otimes m}]}{\|W_m\|_\infty} \\ 
            &\leq \dfrac{2}{\|W_m\|_\infty}. 
        \end{aligned}
    \end{equation}
    Consider the channel ensemble consisting of $\Lambda^{(m)}_1$ and $\Lambda^{(m)}_2$ with probability $\tfrac{1}{2}$, denoted by $\{\tfrac{1}{2},\Lambda^{(m)}_i\}_{i=1}^2$. 
    Then, we have 
        \begin{equation}
            \begin{aligned}
            &\max_{\{M_i\}_i} p_{\mathrm{succ}}(\{\tfrac{1}{2},\Lambda^{(m)}_i\}_{i=1}^2,\{M_i\}_i,\rho^{\otimes m}) \\
            &= \dfrac{1}{2}\left(1 + \dfrac{1}{2}\|\Lambda^{(m)}_1(\rho^{\otimes m}) - \Lambda^{(m)}_2(\rho^{\otimes m})\|_1\right) \\ 
            &> \dfrac{1}{2}\left(1 + \dfrac{1}{\|W_m\|_\infty}\right), 
            \end{aligned}
        \end{equation}
        and 
        \begin{equation}
            \begin{aligned}
            & \max_{\{M_i\}_i} p_{\mathrm{succ}}(\{\tfrac{1}{2},\Lambda^{(m)}_i\}_{i=1}^2,\{M_i\}_i,\sigma^{\otimes m})\\
            &= \dfrac{1}{2}\left(1 + \dfrac{1}{2}\max_{\sigma\in\mathcal{F}(\mathcal{H})}\|\Lambda^{(m)}_1(\sigma^{\otimes m}) - \Lambda^{(m)}_2(\sigma^{\otimes m})\|_1\right) \\ 
            &\leq \dfrac{1}{2}\left(1 + \dfrac{1}{\|W_m\|_\infty}\right). 
            \end{aligned}
        \end{equation}
    Therefore, 
    \begin{equation}
        \begin{aligned}
            &\dfrac{\max_{\{M_i\}_i} p_{\mathrm{succ}}(\{\tfrac{1}{2},\Lambda^{(m)}_i\}_{i=1}^2,\{M_i\}_i,\rho^{\otimes m})}{\max_{\{M_i\}_i} p_{\mathrm{succ}}(\{\tfrac{1}{2},\Lambda^{(m)}_i\}_{i=1}^2,\{M_i\}_i,\sigma^{\otimes m})} \\ 
            &> \dfrac{\tfrac{1}{2}\left(1 + \tfrac{1}{\|W_m\|_\infty}\right)}{\tfrac{1}{2}\left(1 + \tfrac{1}{\|W_m\|_\infty}\right)} \\ 
            &= 1. 
            \label{eq:advantage fixed free state}
        \end{aligned}
    \end{equation}
    Since every free state comes with an integer $m$ satisfying \eqref{eq:advantage fixed free state}, we have
    \begin{equation}
        \max_{m = 2,3,\ldots,d}\dfrac{\displaystyle \max_{\{M^{(m)}_i\}_i} p_{\mathrm{succ}}(\{p_i,\Lambda^{(m)}_i\}_i,\{M^{(m)}_i\}_i,\rho^{\otimes m})}{\displaystyle \max_{\{M^{(m)}_i\}_i} p_{\mathrm{succ}}(\{p_i,\Lambda^{(m)}_i\}_i,\{M^{(m)}_i\}_i,\sigma^{\otimes m})} > 1
    \end{equation}
    for arbitrary $\sigma\in\mathcal{F}(\mathcal{H})$. 
    Since this relation holds even for the worst choice of $\sigma$, the statement follows as desired. 
\end{proof}

\section{Worst-case advantage of resources in channel discrimination}
\label{sec:worst-case}
In this section, we give an operational characterization of the generalized robustness in general QRTs without convexity restriction. 

We consider that the set $\mathcal{F}(\mathcal{H})$ of free states can be expressed as a union  
    \begin{equation}
        \mathcal{F}(\mathcal{H}) = \bigcup_{k} \mathcal{F}_k(\mathcal{H}), 
    \end{equation}
where $\mathcal{F}_k(\mathcal{H})$ denotes a closed convex set for all $k$. 
This expression corresponds to the case in which a QRT is characterized by multiple constraints. 
We give examples of discord, coherence with multiple bases, thermodynamics, and multipartite entanglement below. 
Note that apart from such physically well-motivated cases, mathematically, a decomposition of a given set $\mathcal{F}(\mathcal{H})$ into convex subsets is always possible. 

\begin{example}[Discord]
In discord theory, the set of classical-quantum states is usually taken as the set of free states~\cite{ABC2016}. 
This set can be regarded as the union of the set of classical-quantum states with a fixed local basis. 
More concretely, let $\mathcal{H}_1$ and $\mathcal{H}_2$ be finite-dimensional Hilbert spaces. 
The set of classical-quantum states can be written as 
\begin{equation*}
    \bigcup_{\{\ket{a_k}\}_k} \left\{\sum_k p_k \ketbra{a_k}{a_k} \otimes \rho_k: p_k \geq 0,\, \sum_k p_k = 1,\, \rho_k \in \mathcal{D}(\mathcal{H}_2)  \right\}, 
\end{equation*}
where $\{\ket{a_k}\}_k$ is an orthonormal basis of $\mathcal{H}_1$. 
While the set of classical-quantum states is not convex, each set 
\begin{equation*}
    \left\{\sum_k p_k \ketbra{a_k}{a_k} \otimes \rho_k: p_k \geq 0,\, \sum_k p_k = 1,\, \rho_k \in \mathcal{D}(\mathcal{H}_2)  \right\}
\end{equation*}
is convex by definition. 
The same observation holds for the set of classical-classical states as well. 
\end{example} 
\begin{example}[Coherence in multiple bases]
\label{example:coherence}
In 3D magnetic field sensing \cite{Baumgratz2016}, one would need states with coherence in a specific set of bases, \textit{e.g.}, $x$, $y$, and $z$ bases for a qubit. 
In this case, the set of free states should be the union of the set of incoherent states with respect to the given bases, namely, 
\begin{equation}
    \bigcup_{\alpha = x,y,z} \left\{\rho: \text{$\rho$ is diagonal in $\alpha$ basis} \right\},
\end{equation}
covering the three main axes of the Bloch sphere, while any other state would be regarded as a resource. {\blue We analyze this scenario in detail in Sec.~\ref{subsec:example_coherence_multi-bases}.}
\end{example}

\begin{example}[Thermodynamics]
    In another context, thermodynamical machines such as engines or refrigerators operating between different thermal baths can be characterized by a free set given by the union of the convex subsets corresponding to each equilibrium temperature \cite{Correa2014, Korzekwa2023}. 

    For example, consider two thermal baths with the same Hamiltonian $H$ at inverse temperatures $\beta_1$ and $\beta_2$, respectively. 
    The thermal state corresponding to each thermal bath is given as $\gamma_{\beta_i} \coloneqq \mathrm{e}^{-\beta_i H} / \tr[\mathrm{e}^{-\beta_i H}]$ for $i=1,2$. 
    The set of free states should be given as 
    \begin{equation}
        \{\gamma_{\beta_1} \} \cup \{\gamma_{\beta_2} \}. 
    \end{equation}
\end{example}

\begin{example}[Multipartite Entanglement]
    In the context of multipartite entanglement, multipartite states in the union of the sets of partition-separable states with respect to different bipartitions cannot have genuine multipartite entanglement~\cite{Yamasaki2022activationofgenuine,Palazuelos2022genuinemultipartite}. 
    For any state out of this non-convex set, one can obtain genuine multipartite entanglement from many copies of the state, which is known as genuine multipartite entanglement activation and characterizes this non-convex set from a resource perspective~\cite{Yamasaki2022activationofgenuine,Palazuelos2022genuinemultipartite}. 
\end{example}

We introduce a concept of \textit{worst-case} advantage of channel discrimination in general QRTs without convexity restriction. 
This notion of worst-case advantage arises when a QRT has multiple constraints as introduced above. If we do not know which constraints we eventually focus on, we have to take the worst-case scenario into account. 
Here, we consider conventional channel discrimination for single copies, where $m=1$ in the $m$-input channel discrimination.  
Instead of considering the maximum success probability, we fix the measurement strategy in this case; that is, when we quantify an advantage of a resource state, we use the same measurement for the resource state and the free states. 
Let $\{p_i,\Lambda_i\}_i$ be a channel ensemble, and $\{M_i\}_i$ be a given measurement. 
Then, the advantage of a resource state $\rho$ with respect to $\mathcal{F}_k(\mathcal{H})$ can be defined as 
\begin{equation}
    \label{eq:k-advantage}
    \dfrac{p_{\mathrm{succ}}(\rho,\{p_i,\Lambda_i\}_i,\{M_i\}_i)}{\max_{\sigma_k \in \mathcal{F}_k(\mathcal{H})} p_{\mathrm{succ}}(\sigma_k,\{p_i,\Lambda_i\}_i,\{M_i\}_i)}, 
\end{equation}
where $p_{\mathrm{succ}}$ denotes the success probability, which is obtained by taking $m = 1$ in Eq.~\eqref{eq:success_prob}. 
We consider the worst-case advantage of $\rho$ to be the infimum of Eq.~\eqref{eq:k-advantage} over $k$.  

We find that the worst-case maximum advantage for channel discrimination is quantified by the generalized robustness. 
\begin{theorem}[Theorem~4 in the companion Letter~\cite{PRL}]\label{thm:advantage worst case}
    Let $\mathcal{H}$ be a $d$-dimensional Hilbert space. 
    For any resource state $\rho \in \mathcal{F}(\mathcal{H})\backslash \mathcal{D}(\mathcal{H})$, 
    \begin{equation}
        \label{eq:worst_case_advantage_robustness}
        \begin{aligned}
            &\inf_{k} \max_{\{p_i,\Lambda_i\}_i,\{M_i\}_i} \dfrac{p_{\mathrm{succ}}(\rho,\{p_i,\Lambda_i\}_i,\{M_i\}_i)}{\max_{\sigma_k \in \mathcal{F}_k(\mathcal{H})} p_{\mathrm{succ}}(\sigma_k,\{p_i,\Lambda_i\}_i,\{M_i\}_i)} \\  
            &= 1 + R_{\mathcal{F}(\mathcal{H})}(\rho). 
        \end{aligned}
    \end{equation}
\end{theorem}
Thus, the generalized robustness can be interpreted as the worst-case maximum advantage under multiple sets of free states.
{\blue
Since the generalized robustness is faithful,~\eqref{eq:worst_case_advantage_robustness} implies that every resource state is useful in the sense of the worst-case maximum advantage for channel discrimination. 
}
Mathematically, this result does not depend on which decomposition of $\mathcal{F}(\mathcal{H})$ we pick. 
We may apply this result to whatever scenario we want to consider. 
When the set $\mathcal{F}(\mathcal{H})$ of free states happens to be convex, then the infimum over $k$ is omitted, and this result will naturally be reduced to the previous result for convex QRTs by Ref.~\cite{Takagi2019b}. 

\begin{figure}[htbp]
        \centering
        \includegraphics[keepaspectratio, width=\linewidth]{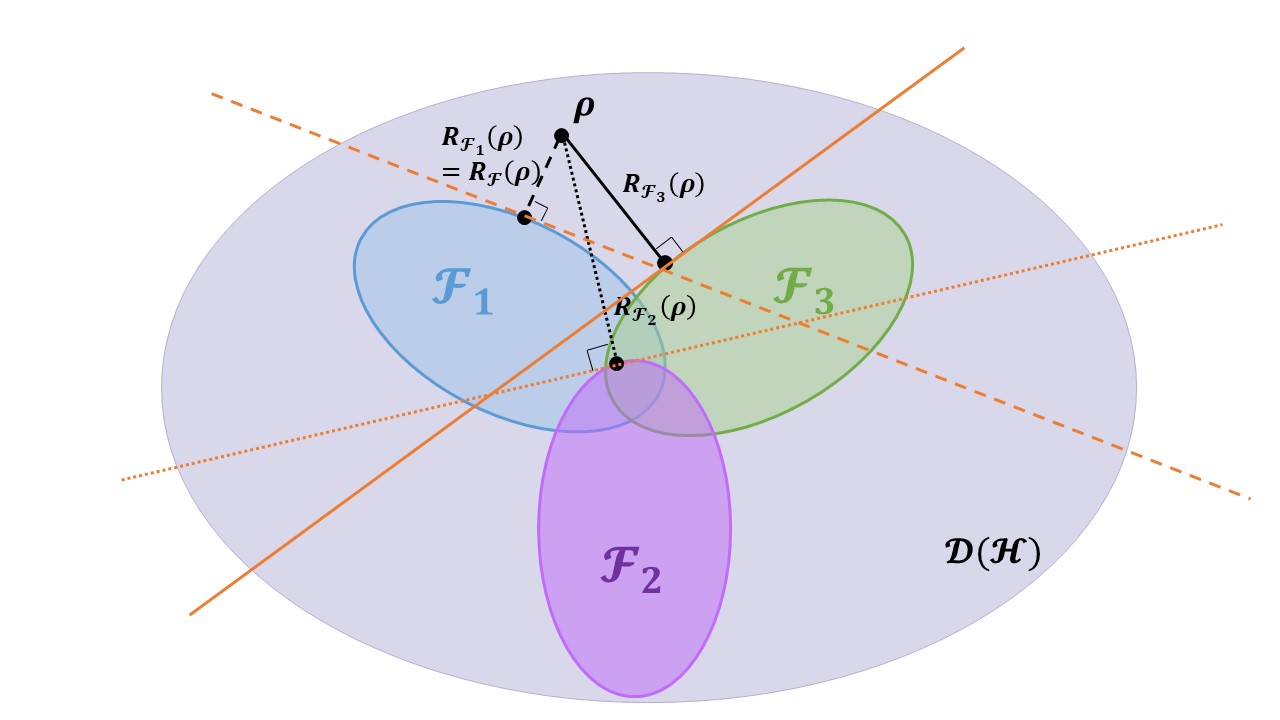}
        \caption{A diagram of the generalized robustness with respect to a non-convex set $\mathcal{F}$ of free states formed by three convex subsets $\mathcal{F}_1$, $\mathcal{F}_2$, and $\mathcal{F}_3$ with $\mathcal{F}=\bigcup_{k=1,2,3}\mathcal{F}_k$.  
        The dashed, dotted, and solid (orange) tangent lines 
        represent the hyperplane separating $\rho$ from $\mathcal{F}_1$, $\mathcal{F}_2$, and $\mathcal{F}_3$, respectively. 
        One can observe that these linear witnesses cannot separate the resource state $\rho$ from the non-convex set $\mathcal{F}$. 
        As illustrated in the figure, the generalized robustness $R_{\mathcal{F}}(\rho)$ is given as the minimum value among $R_{\mathcal{F}_1}(\rho)$, $R_{\mathcal{F}_2}(\rho)$, and $R_{\mathcal{F}_3}(\rho)$. 
        Since $\mathcal{F}_1$ is the closest to $\rho$ in this case, the value of $R_{\mathcal{F}_1}(\rho)$ will be employed as $R_{\mathcal{F}}(\rho)$. 
        }
        \label{fig:worst_case}
\end{figure}

To prove this theorem, we use the following lemma,  
showing that the generalized robustness with respect to the union of convex subsets is given by the infimum value of the generalized robustness with respect to each subset. 
We also show a diagram for an intuitive understanding of this lemma in FIG.~\ref{fig:worst_case}.

\begin{lemma}~\label{lem:worst_case_robustness}
    For any state $\rho \in \mathcal{D}(\mathcal{H})$, 
    \begin{equation}
        R_{\mathcal{F}(\mathcal{H})}(\rho) = \inf_{k} R_{\mathcal{F}_k(\mathcal{H})}(\rho). 
    \end{equation}
\end{lemma}
\begin{proof}
    Since $\mathcal{F}_k(\mathcal{H}) \subseteq \mathcal{F}(\mathcal{H})$ for all $k$, we have $R_{\mathcal{F}(\mathcal{H})}(\rho) \leq R_{\mathcal{F}_k(\mathcal{H})}(\rho)$ for all $k$ by the definition of the generalized robustness. 
    By taking the infimum over $k$ on both sides, we have 
    \begin{equation}
        R_{\mathcal{F}(\mathcal{H})}(\rho) \leq \inf_{k} R_{\mathcal{F}_k(\mathcal{H})}(\rho)
    \end{equation}
    On the other hand, let $\tau \in \mathcal{D}(\mathcal{H})$ be a free state achieving $R_{\mathcal{F}(\mathcal{H})}(\rho)$; that is, 
    \begin{equation}
        \dfrac{\rho + R_{\mathcal{F}(\mathcal{H})}(\rho)\tau}{1 + R_{\mathcal{F}(\mathcal{H})}(\rho)} \in \mathcal{F}(\mathcal{H}). 
    \end{equation}
    Since $\mathcal{F}(\mathcal{H}) = \bigcup_k \mathcal{F}_k(\mathcal{H})$, there exists $k$ such that 
    \begin{equation}
        \dfrac{\rho + R_{\mathcal{F}(\mathcal{H})}(\rho)\tau}{1 + R_{\mathcal{F}(\mathcal{H})}(\rho)} \in \mathcal{F}_k(\mathcal{H}). 
    \end{equation}
    This shows that $R_{\mathcal{F}(\mathcal{H})}(\rho)$ serves as a suboptimal value of $R_{\mathcal{F}_k(\mathcal{H})}(\rho)$. 
    Hence, 
    \begin{equation}
        R_{\mathcal{F}(\mathcal{H})}(\rho) 
        \geq R_{\mathcal{F}_k(\mathcal{H})}(\rho) \geq \inf_k R_{\mathcal{F}_k(\mathcal{H})}(\rho). 
    \end{equation}
\end{proof}
With this characterization, we can prove Theorem~\ref{thm:advantage worst case}. 
\begin{proof}[Proof of Theorem~\ref{thm:advantage worst case}]
For any $k$, since $\mathcal{F}_k(\mathcal{H})$ is convex, we can apply the analysis from Ref.~\cite{Takagi2019b} and get 
\begin{equation}
    \begin{aligned}
        &\max_{\{p_i,\Lambda_i\}_i,\{M_i\}_i} \dfrac{p_{\mathrm{succ}}(\rho,\{p_i,\Lambda_i\}_i,\{M_i\}_i)}{\max_{\sigma_k \in \mathcal{F}_k(\mathcal{H})} p_{\mathrm{succ}}(\sigma_k,\{p_i,\Lambda_i\}_i,\{M_i\}_i)}  \\ 
        &= 1 + R_{\mathcal{F}_k(\mathcal{H})}(\rho). 
    \end{aligned}
\end{equation}
Taking the infimum over $k$ on both sides, we have 
\begin{equation}
    \begin{aligned}
        &\inf_{k} \max_{\{p_i,\Lambda_i\}_i,\{M_i\}_i} \dfrac{p_{\mathrm{succ}}(\rho,\{p_i,\Lambda_i\}_i,\{M_i\}_i)}{\max_{\sigma_k \in \mathcal{F}_k(\mathcal{H})} p_{\mathrm{succ}}(\sigma_k,\{p_i,\Lambda_i\}_i,\{M_i\}_i)} \\
        &= 1 + \inf_{k} R_{\mathcal{F}_k(\mathcal{H})}(\rho).
    \end{aligned} 
\end{equation}
By Lemma~\ref{lem:worst_case_robustness}, we know $\inf_{k} R_{\mathcal{F}_k(\mathcal{H})}(\rho) = R_{\mathcal{F}(\mathcal{H})}(\rho)$. 
Hence, 
    \begin{equation}
        \begin{aligned}
            &\inf_{k} \max_{\{p_i,\Lambda_i\}_i,\{M_i\}_i} \dfrac{p_{\mathrm{succ}}(\rho,\{p_i,\Lambda_i\}_i,\{M_i\}_i)}{\max_{\sigma_k \in \mathcal{F}_k(\mathcal{H})} p_{\mathrm{succ}}(\sigma_k,\{p_i,\Lambda_i\}_i,\{M_i\}_i)} \\  
            &= 1 + R_{\mathcal{F}(\mathcal{H})}(\rho). 
        \end{aligned}
    \end{equation}
as desired. 
\end{proof}

\begin{figure*}[t]
    \centering 
    \includegraphics[width = 6.75in]{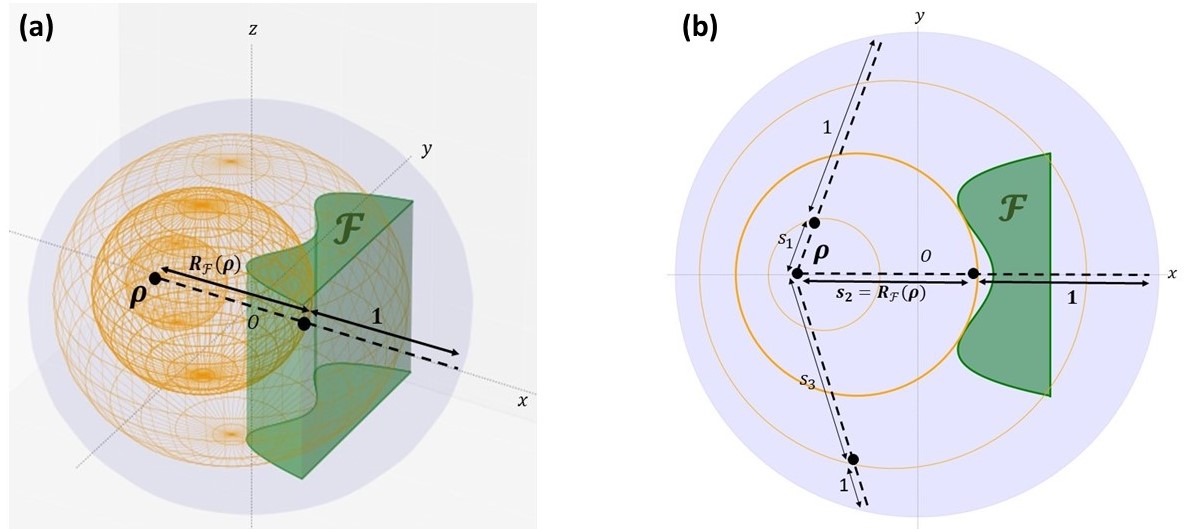}
    \caption{(a) A 3D diagram depicting curved surfaces representing the witness for the single-qubit case, i.e., $d = 2$. 
    In the diagram, the outermost blue ball represents the Bloch ball, that is, the set of single-qubit states. 
    The green non-convex shape represents the set $\mathcal{F}$ of free states. 
    Three orange meshed spheres correspond to curved surfaces $S_{2,\rho,s} = 0$ for different values of the parameter $s$, each of which is a set of states that internally divide the state $\rho$ and the Bloch sphere into the ratio $s:1$. 
    The middle orange sphere touches $\mathcal{F}$, and the parameter $s$ for this sphere corresponds to the generalized robustness $R_{\mathcal{F}}(\rho)$.
    (b) A diagram of the 2D plane obtained by cutting (a) at $z = 0$.
    The three orange circles corresponds to $S_{2,\rho,s} = 0$ with $s = s_1,s_2,s_3$, respectively. 
    The middle circle $S_{2,\rho,s_2} = 0$ touches $\mathcal{F}$, and thus $s_2 = R_{\mathcal{F}}(\rho)$.
    }
    \label{fig:robustness_witness}
\end{figure*}

{\blue \section{Example: Single-qubit case}
\label{sec:special_case_qubit}
In this section, we describe an illustrative example of our methods and results for single-qubit QRTs.}
In the single-qubit case, i.e., $d=2$, the construction of the multi-copy witness is largely simplified due to the clear geometric structure of the set of quantum states, and we here explicitly derive the description of the witness. 
{\blue 
In Sec.~\ref{subsec:general_derivation_single_qubit}, we show a general derivation of a $2$-copy witness for single-qubit QRTs, as shown in FIG.~\ref{fig:robustness_witness}.
In Sec.~\ref{subsec:example_coherence_multi-bases}, we give a more specific example of the construction of our multi-copy witness for a physically motivated choice of a single-qubit non-convex QRT, namely, the case of coherence in $x$, $y$, and $z$ bases introduced in Example~\ref{example:coherence} of Sec.~\ref{sec:worst-case}. 
}

{\blue
\subsection{General derivation of 2-copy witness for single-qubit QRTs}
\label{subsec:general_derivation_single_qubit}
In the case of single-qubit QRTs ($d=2$),
}
we have only to consider  
\begin{equation}
    \begin{aligned}
    S_2(\eta) 
    &= \frac{1}{2}\left(\tr[\eta]S_1(\eta) - \tr[\eta^2]S_0(\eta) \right) \\ 
    &= \frac{1}{2}\left(\tr[\eta]^2 - \tr[\eta^2] \right)\\
    &= \frac{1 - x_1^2 - x_2^2 - x_3^2}{4}
    \end{aligned}
\end{equation}
for Hermitian operator $\eta$ with Bloch vector $(x_1,x_2,x_3)$. 
One may observe that $S_2(\eta) \geq 0$ fully characterizes the Bloch ball, and $S_2(\eta) = 0$ corresponds to the surface of the Bloch ball (i.e., the Bloch sphere). 

Let us fix a single-qubit quantum state $\rho$ with Bloch vector $(r_1,r_2,r_3)$. 
Then, for any state $\eta$ with Bloch vector $(x_1,x_2,x_3)$, we have 
\begin{equation}
    \begin{aligned}
        S_{2,\rho,s}(\eta) 
        &= \frac{1}{4} - \frac{r_1^2+r_2^2+r_3^2}{4s^2} + \dfrac{(1+s)(r_1x_1 + r_2x_2 + r_3x_3)}{2s^2}\\ 
        & - \dfrac{(1+s)^2(x_1^2 + x_2^2 + x_3^2)}{4s^2}. 
    \end{aligned}
\end{equation}

We here describe how our results for the general $d$-dimensional case look like in this single-qubit case. 
We define  
    \begin{equation}
    \label{eq:quadratic_witness_qubit}
    \begin{aligned}
    W(\rho,s)
        &\coloneqq \left(1-\dfrac{r_1^2+r_2^2+r_3^2}{s^2}\right)I_2\otimes I_2 \\
    &- \dfrac{(1+s)^2}{s^2}\left[\sigma_1\otimes\sigma_1 + \sigma_2\otimes\sigma_2 + \sigma_3\otimes\sigma_3\right]\\
    &+\dfrac{1+s}{s^2}[r_1\left(I_2\otimes \sigma_1 + \sigma_1\otimes I_2\right) \\
    &\quad\quad\quad+r_2\left(I_2\otimes \sigma_2 + \sigma_2\otimes I_2\right) \\ 
    &\quad\quad\quad+r_3\left(I_2\otimes \sigma_3 + \sigma_3\otimes I_2\right)], 
    \end{aligned}
\end{equation}
with the 1-qubit identity operator $I_2$ and Pauli operators 
    \begin{equation}
        \sigma_1 \coloneqq \left(
        \begin{array}{cc}
            0 & 1 \\
            1 & 0
        \end{array}
        \right),\,\, 
        \sigma_2 \coloneqq \left(
        \begin{array}{cc}
            0 & -i \\
            i & 0
        \end{array}
        \right),\,\,
        \sigma_3 \coloneqq \left(
        \begin{array}{cc}
            1 & 0 \\
            0 & -1
        \end{array}
        \right). 
    \end{equation}
    Then, for a state 
    $\eta = \tfrac{1}{2}\left(I_2 + x_1\sigma_1 + x_2\sigma_2 + x_3\sigma_3\right)$    
    with Bloch vector $(x_1,x_2,x_3)$, 
    we have 
    \begin{equation}
    \begin{aligned}
    \tr\left[W(\rho,s)\eta^{\otimes2}\right] = S_2\left(\frac{1+s}{s}\eta - \frac{1}{s}\rho \right) = S_{2,\rho,s}(\eta),  
    \end{aligned}
    \end{equation} 
    which is shown in Appendix~\ref{app:witness_derivation_qubit} in detail.
    
    We show illustrative diagrams for the single-qubit case in FIG.~\ref{fig:robustness_witness}.
    As shown in FIG.~\ref{fig:robustness_witness}, 
    for each $s > 0$, the curve $S_{2,\rho,s}(\eta) = 0$ represents a three-dimensional sphere in the Bloch ball, and when $s < R_{\mathcal{F}(\mathcal{H})}$, this sphere separates resource state $\rho$ from the set $\mathcal{F}$ of free states. 
    In this way, the description here makes it possible to interpret our multi-copy ($2$-copy) witness geometrically in the Bloch ball in the single-qubit case. 
    Note that the corresponding interpretation may not be straightforward in the cases of higher dimensions, i.e., $d>2$ in general, due to the lack of suitable high-dimensional generalizations of the Bloch ball. 

{\blue 
\subsection{Multi-copy witness and generalized robustness for QRT of multi-basis coherence}
\label{subsec:example_coherence_multi-bases}

Here we give an explicit example of a multi-copy witness as constructed in Sec.~\ref{subsec:general_derivation_single_qubit}. 
We consider the QRT of multi-basis coherence, where we want coherent single-qubit states with respect to $x$, $y$, and $z$ bases; 
refer to  Example~\ref{example:coherence} in Sec.~\ref{sec:worst-case} for physical interpretation and motivation of this setup. 
This example concisely illustrates how our multi-copy witnesses are constructed in a relevant single-qubit non-convex QRT. 
Note that while the construction of a multi-copy witness in a general $d$-dimensional QRT also proceeds in a similar manner, the computation may be more complicated in such a general case due to the lack of the simple geometric interpretation given in Sec.~\ref{subsec:general_derivation_single_qubit}. 

Let $\mathcal{H}$ be a two-dimensional Hilbert space. 
For $\alpha = x, y, z$, define the set of incoherent states in $\alpha$ basis as 
\begin{equation}
    \begin{aligned}
        \mathcal{F}_{\alpha}(\mathcal{H}) \coloneqq \Bigg\{
        &\sigma \in \mathcal{D}(\mathcal{H}): \\ 
        &\sigma = p\ketbra{+_\alpha}{+_{\alpha}} + (1-p)\ketbra{-_\alpha}{-_{\alpha}},\\ 
        &0\leq p\leq 1 \Bigg\}, 
    \end{aligned}
\end{equation}
where $\{\ket{+_{\alpha}}, \ket{-_{\alpha}}\}$ denotes the orthonormal $\alpha$ basis, namely, 
\begin{alignat}{3}
    &\ket{+_x} \coloneqq \frac{\ket{0}+\ket{1}}{\sqrt{2}},&\quad\quad &\ket{-_x} \coloneqq \frac{\ket{0}-\ket{1}}{\sqrt{2}},& \\ 
    &\ket{+_y} \coloneqq \frac{\ket{0}+i\ket{1}}{\sqrt{2}},& \quad\quad&\ket{-_y} \coloneqq \frac{\ket{0}-i\ket{1}}{\sqrt{2}},&\\
    &\ket{+_z} \coloneqq \ket{0},& \quad\quad &\ket{-_z} \coloneqq \ket{1},&
\end{alignat}
with the computational basis $\{\ket{0},\ket{1}\}$. 
Define the set of free states as 
\begin{equation}
    \mathcal{F}(\mathcal{H}) = \bigcup_{\alpha = x,y,z} \mathcal{F}(\mathcal{H}_{\alpha}). 
\end{equation}
Suppose that we are given a resource state $\rho \in \mathcal{D}(\mathcal{H})\backslash \mathcal{F}(\mathcal{H})$ with the Bloch vector $(r_1,r_2,r_3)$. 
The condition that $\rho$ is a resource state gives a constraint on the Bloch vector that at least two of $r_1,r_2,r_3$ are not zero. 

In Sec.~\ref{subsec:general_derivation_single_qubit}, we already constructed the Hermitian operator $W(\rho,s)$ in~\eqref{eq:quadratic_witness_qubit}, which serves as a $2$-copy witness for {\blue $s < R_{\mathcal{F}(\mathcal{H})}$}. 
To see when $W(\rho,s)$ can discern $\rho$ from $\mathcal{F}(\mathcal{H})$ in the QRT of multi-basis coherence, 
here we compute the generalized robustness of $\rho$. 

For the computation of $R_{\mathcal{F}(\mathcal{H})}$, we use our analysis in Sec.~\ref{sec:worst-case}. In particular, using Lemma~\ref{lem:worst_case_robustness}, we have  
\begin{equation}
    \label{eq:robustness_multi_basis_coherence}
    R_{\mathcal{F}(\mathcal{H})}(\rho) = \min \Big\{R_{F_\alpha(\mathcal{H})}(\rho):\alpha = x,y,z\Big\}. 
\end{equation}
Hence, we can compute $R_{\mathcal{F}(\mathcal{H})}(\rho)$ by independently computing $R_{\mathcal{F}_\alpha(\mathcal{H})}$ for $\alpha = x,y,z$ and then taking the minimum. 
Let us focus on $\mathcal{F}_{x}(\mathcal{H})$, as the other cases proceed similarly.  
Recall that the region $S_{2,\rho,s}(\eta) \geq 0$ is expressed as  
\begin{equation}
    \left(\dfrac{1+s}{s}x_1 - \dfrac{1}{s}r_1\right)^2 + \left(\dfrac{1+s}{s}x_2 - \dfrac{1}{s}r_2\right)^2 + \left(\dfrac{1+s}{s}x_3 - \dfrac{1}{s}r_3\right)^2 \leq 1. 
\end{equation}
Considering the geometric analysis in Sec.~\ref{subsec:general_derivation_single_qubit}, we have $s = R_{\mathcal{F}_x(\mathcal{H})}(\rho)$ if and only if the surface of the region touches the set of free states, which is expressed by 
\begin{equation}
    \left\{(x_1,0,0): -1\leq x_1 \leq 1\right\}. 
\end{equation}
Hence, when $s = R_{\mathcal{F}_x(\mathcal{H})}(\rho)$, the equation 
\begin{equation}
    \left(\dfrac{1+s}{s}x_1 - \dfrac{1}{s}r_1\right)^2 + \dfrac{r_2^2}{s^2} + \dfrac{r_3^2}{s^2} = 1
\end{equation}
with respect to $x_1$ has a solution with multiplicity $2$. 
This occurs if and only if 
\begin{equation}
     \dfrac{r_2^2}{s^2} + \dfrac{r_3^2}{s^2} = 1; 
\end{equation}
that is, since $s > 0$,  
\begin{equation}
    s = \sqrt{r_2^2 + r_3^2}. 
\end{equation}
Hence, we have 
\begin{equation}
    R_{\mathcal{F}_{x}(\mathcal{H})}(\rho) = \sqrt{r_2^2 + r_3^2}. 
\end{equation}
Similarly, 
\begin{align}
    R_{\mathcal{F}_{y}(\mathcal{H})}(\rho) &= \sqrt{r_3^2 + r_1^2}, \\ 
    R_{\mathcal{F}_{z}(\mathcal{H})}(\rho) &= \sqrt{r_1^2 + r_2^2}. 
\end{align}
Thus, from~\eqref{eq:robustness_multi_basis_coherence}, 
\begin{equation}
    R_{\mathcal{F}(\mathcal{H})}(\rho) = \min\left\{\sqrt{r_2^2 + r_3^2}, \sqrt{r_3^2 + r_1^2}, \sqrt{r_1^2 + r_2^2} \right\}. 
\end{equation}
Without loss of generality, assume that $R_{\mathcal{F}(\mathcal{H})}(\rho) = \sqrt{r_1^2 + r_2^2}$. 
Then, with $s = \sqrt{r_1^2 + r_2^2}$, $W(\rho,s)$ constructed in Theorem~\ref{thm:high_order_witness_qudit}  is now given as 
\begin{equation}
    \begin{aligned}
    W(\rho, s)
        &\coloneqq \left(\dfrac{r_1^2+r_2^2+r_3^2}{r_1^2 + r_2^2} - 1\right)I_2\otimes I_2 \\ 
        &+ \dfrac{\left(1+\sqrt{r_1^2 + r_2^2}\right)^2}{r_1^2 + r_2^2}\left[\sigma_1\otimes\sigma_1 + \sigma_2\otimes\sigma_2 + \sigma_3\otimes\sigma_3\right]\\
        &-\dfrac{1+\sqrt{r_1^2 + r_2^2}}{r_1^2 + r_2^2}\Bigg[r_1\left(I_2\otimes \sigma_1 + \sigma_1\otimes I_2\right) \\ 
        &+r_2\left(I_2\otimes \sigma_2 + \sigma_2\otimes I_2\right) 
        +r_3\left(I_2\otimes \sigma_3 + \sigma_3\otimes I_2\right)\Bigg].  
    \end{aligned}
\end{equation}
}

\section{Analysis of weight-based measure}
\label{sec:extension_weight}
In this section, we characterize another resource measure, the weight of resource \cite{Lewenstein1998,Skrzypczyk2014,Pusey2015,Cavalcanti2016,Kaifeng2018,Ducuara2020,Uola2020}, which was also conventionally investigated in convex QRTs. 
In Sec.~\ref{subsec:construction_multi-copy_witness_weight}, we give another construction of multi-copy witness based on the weight of resource. 
In Sec.~\ref{subsec:operational_advantage_weight}, we show that every quantum state is useful for some channel exclusion task based on the multi-copy witnesses that we construct in Sec.~\ref{subsec:construction_multi-copy_witness_weight}. 
In Sec.~\ref{subsec:worstcase_advantage_weight}, we show that the weight of resource can be operationally interpreted as the worst-case advantage in channel exclusion. 

Here, we recall the definition of the weight of resource. 
\begin{definition}
    Let $\rho \in \mathcal{D}(\mathcal{H})$ be a quantum state. Then, the weight of resource $\mathrm{WoR}_{\mathcal{F}(\mathcal{H})}(\rho)$ of $\rho$ with respect to the set $\mathcal{F}(\mathcal{H})$ of free states is defined by 
    \begin{equation}
        \begin{aligned}
        \label{eq:weight}
        \mathrm{WoR}_{\mathcal{F}(\mathcal{H})}(\rho) = \min_{\sigma \in \mathcal{F}(\mathcal{H}), \tau \in \mathcal{D}(\mathcal{H})}\Bigg\{s\geq 0 : \
        &\rho = s\tau + (1-s)\sigma\Bigg\}. 
        \end{aligned}
    \end{equation}
\end{definition}
Although the form in the definition resembles the one in the generalized robustness, these two resource measures behave differently and provide distinct operational constraints in the context of resource distillation~\cite{Regula2021_fundamentallimitation}.
Note that in the definition of the weight of resource, we may always pick $s = 1$ by taking $\tau = \rho$. 
Thus, for all states $\rho$, $\mathrm{WoR}_{\mathcal{F}(\mathcal{H})}(\rho) \leq 1$. 

\subsection{Another construction of multi-copy witness}
\label{subsec:construction_multi-copy_witness_weight}
In this section, we give another construction of multi-copy witnesses in QRTs without convexity restriction, based on the weight of resource. 
For the construction, we use the following lemma, which characterizes quantum states inspired by the definition of the weight of resource. 
We give the proof in Appendix~\ref{app:proof_witness_characterization_weight}. 
\begin{lemma}~\label{lem:witness_characterization_weight}
    Let $\rho$ be a quantum state, and let $0 < s < 1$. 
    Consider a quantum state $\eta$.  
    Then, 
    $\eta$ can be expressed as 
    \begin{equation}
        \eta = \frac{\rho - s' \tau}{1 - s'}
    \end{equation}
    by using some positive number $0 < s' \leq s$ and some state $\tau \in \mathcal{D}(\mathcal{H})$ 
    if and only if it holds for all $m = 1,2,\ldots,d$ that
    \begin{equation}
        S_m\left(\frac{1}{s}\rho - \frac{1-s}{s} \eta\right) \geq 0,
    \end{equation}
     with $S_m$ given in Lemma~\ref{lem:bloch_ball}. 
\end{lemma}

Now, we show the construction of our multi-copy witnesses. 
\begin{theorem}
    \label{thm:high_order_witness_qudit_weight}
     Let $\mathcal{H}$ be a $d$-dimensional Hilbert space. 
    Let $\rho \in \mathcal{D}(\mathcal{H})\backslash\mathcal{F}(\mathcal{H})$ be a resource state.  
    Then, for $0 < s < 1$, we can construct a family $\widetilde{\mathcal{W}}^{(\mathrm{WoR})}_{\rho,s} \coloneqq (\widetilde{W}^{(\mathrm{WoR})}_{m}(\rho,s) \in \Herm(\mathcal{H}^{\otimes m}): m = 2,3,\ldots,d)$ of Hermitian operators with the following property:   
    \begin{align}
    \label{eq:conventional_witness_resource_weight}
    &\max_{m = 2,3,\ldots,d}\tr\left[\widetilde{W}^{(\mathrm{WoR})}_{m}(\rho,s)\rho^{\otimes m}\right] < 0, \\ 
    \label{eq:conventional_witness_free_weight}
    &\max_{m = 2,3,\ldots,d} \tr\left[\widetilde{W}^{(\mathrm{WoR})}_{m}(\rho,s)\sigma^{\otimes m}\right] \geq 0, \,\,\forall \sigma \in \mathcal{F}(\mathcal{H})
\end{align}
    if and only if
    \begin{equation}
        s < \mathrm{WoR}_{\mathcal{F}(\mathcal{H})}(\rho).
    \end{equation}
\end{theorem}
\begin{proof}

For a state $\eta \in \mathcal{D}(\mathcal{H})$, define 
\begin{equation}
    \widetilde{S}_{m,\rho,s}(\eta) \coloneqq S_m\left(\frac{1}{s}\rho - \frac{1-s}{s} \eta \right) 
\end{equation}
for $m = 1,2,\ldots,d$.
Note that 
\begin{equation}
    \widetilde{S}_{1,\rho,s}(\eta) =  \frac{1}{s}\tr[\rho] - \frac{1-s}{s}\tr[\eta] = 1\geq 0, 
\end{equation}
and thus, we will focus on $m=2,3,\ldots,d$ in the following.

It trivially follows that $\widetilde{S}_{m,\rho,s}(\rho) = S_m(\rho) \geq 0$ for all $m$ since $\rho$ is a state, in particular, a positive semidefinite operator. 
When $s < R_{\mathcal{F}(\mathcal{H})}(\rho)$, by the definition of the weight of resource, for any free state $\sigma \in \mathcal{F}(\mathcal{H})$, there is no positive number $s' \leq  s$ such that 
\begin{equation}
    \sigma = \frac{\rho - s'\tau}{1-s'}
\end{equation}
with some $\tau \in \mathcal{D}(\mathcal{H})$.
By Lemma~\ref{lem:witness_characterization_weight}, in this case, there exists $m\in\{2,\ldots,d\}$ such that $\widetilde{S}_{m,\rho,s}(\sigma) < 0$
On the other hand, when $s \geq \mathrm{WoR}_{\mathcal{F}(\mathcal{H})}(\rho)$, there exists $0 \leq s' \leq s$ such that 
\begin{equation}
    \sigma = \frac{\rho - s'\tau}{1-s'}
\end{equation}
with some $\tau \in \mathcal{D}(\mathcal{H})$.
Indeed, we can take a free state $\sigma$ achieving $\mathrm{WoR}_{\mathcal{F}(\mathcal{H})}(\rho)$ with $s' = \mathrm{WoR}_{\mathcal{F}(\mathcal{H})}(\rho)$. 
In this case, for all $m = 2,3,\ldots, d$, we have 
$\widetilde{S}_{m,\rho,s}(\sigma)  \geq 0$ for such $\sigma$.
With a similar procedure as in the proof of Proposition~\ref{thm:high_order_witness_qudit}, for all $2\leq m\leq d$ and $s>0$, 
we can construct $W_m(\rho,s)$ satisfying 
\begin{equation}
    \tr\left[W^{(\mathrm{WoR})}_{m}(\rho,s) \eta^{\otimes n}\right] = \widetilde{S}_{m,\rho,s}(\eta)
\end{equation}
for $\eta \in \mathcal{D}(\mathcal{H})$. 

Now, based on $W^{(\mathrm{WoR})}_m(\rho,s)$, we define 
\begin{equation}
    \widetilde{W}^{(\mathrm{WoR})}_{m}(\rho,s) \coloneqq -C\left(W(\rho,s) + \Delta^{(\mathrm{WoR})}_{m}(\rho,s) I^{\otimes m}\right), 
\end{equation}
with 
\begin{equation}
    \Delta^{(\mathrm{WoR})}_{m}(\rho,s) \coloneqq \min_{\sigma \in \mathcal{F}(\mathcal{H})} 
    \left|\tr[W^{(\mathrm{WoR})}_m(\rho,s)\sigma^{\otimes m}]\right| (> 0), 
\end{equation}
where $C>0$ is any suitable constant for normalization. 
With the same argument used to show Theorem~\ref{cor:high_order_witness_qudit_conventional} from Proposition~\ref{thm:high_order_witness_qudit}, this construction satisies Eqs.~\eqref{eq:conventional_witness_resource_weight} and \eqref{eq:conventional_witness_free_weight}. 
\end{proof}

\subsection{Operational advantage in multi-input channel exclusion task}
\label{subsec:operational_advantage_weight}
In this section, we show that every resource is useful in a task called {\em channel exclusion} \cite{Ducuara2020,Uola2020}, based on the multi-copy witnesses constructed in the previous section. 

We consider an $m$-input channel exclusion task. 
In the $m$-input channel exclusion task, one aims to choose a channel that did not take place. 
Let $\rho \in \mathcal{D}(\mathcal{H})$ be a state. 
Let $\{p_i,\Lambda_i^{(m)}\}_i$ be a channel ensemble, where $\Lambda_i^{(m)}$ are channels acting on $\mathcal{D}(\mathcal{H}^{\otimes m})$. 
We are given a state $\Lambda^{(m)}(\rho^{\otimes m})$, where $\Lambda_i^{(m)}$ is randomly sampled with probability $p_i$ from the ensemble.
The goal of this task is to choose a channel that was not applied to $\rho^{\otimes m}$. 
For this purpose, we perform a quantum measurement $\{M_i\}_i$, where the measurement result $i$ indicates that the channel $\Lambda^{(m)}_i$ took place. 
The error probability of this task is defined by 
\begin{equation}
    p_{\mathrm{err}}(\rho,\{p_i,\Lambda^{(m)}_i\}_i,\{M_i\}_i) \coloneqq \sum_{i} p_i \tr[M_i \Lambda^{(m)}(\rho^{\otimes m})]. 
\end{equation}
In the task, we would like to minimize the error probability by appropriately choosing a measurement strategy, so we consider 
\begin{equation}
    \min_{\{M_i\}_i} p_{\mathrm{err}}(\rho,\{p_i,\Lambda^{(m)}_i\}_i,\{M_i\}_i). 
\end{equation}
We show that for any resource state $\rho \in \mathcal{D}(\mathcal{H})$, we can design an $m$-input channel exclusion task based on the multi-copy witnesses constructed in the previous section so that $\rho$ has an advantage over any free states. 

\begin{theorem}
\label{thm:advantage_exclusion}
Let $\mathcal{H}$ be a $d$-dimensional Hilbert space. 
For any resource state $\rho \in \mathcal{D}(\mathcal{H})\backslash\mathcal{F}(\mathcal{H})$, there exists a family of channel ensembles $\left(\{p_i,\Lambda^{(m)}_i\}_i\right)_{m=2}^{d}$ such that 
\begin{equation}
    \begin{aligned}
        \max_{\sigma \in \mathcal{F}(\mathcal{H})} \min_{m = 2,3,\ldots,d} \frac{\displaystyle \min_{\{M_i\}_i} p_{\mathrm{err}}(\{p_i,\Lambda^{(m)}_i\}_i,\{M_i\}_i,\rho^{\otimes m})}{\displaystyle \min_{\{M_i\}_i} p_{\mathrm{err}}(\{p_i,\Lambda^{(m)}_i\}_i,\{M_i\}_i,\sigma^{\otimes m})}
        < 1. 
    \end{aligned} 
\end{equation}
\end{theorem}
The proof is in a similar line with that of Theorem~\ref{thm:advantage_multi-copy}. 
The detail of the proof is provided in Appendix~\ref{app:proof_advantage_exclusion}. 

\subsection{Worst-case advantage in channel exclusion task}
\label{subsec:worstcase_advantage_weight}
In this section, we give an operational characterization of the weight of resource in general QRTs without convexity restriction in the same way as Sec.~\ref{sec:worst-case}. 

We consider once more that the set $\mathcal{F}(\mathcal{H})$ of free states can be expressed as a union  
    \begin{equation}
        \mathcal{F}(\mathcal{H}) = \bigcup_{k} \mathcal{F}_k(\mathcal{H}), 
    \end{equation}
where $\mathcal{F}_k(\mathcal{H})$ denotes a closed convex set for all $k$. 

We introduce a concept of \textit{worst-case} advantage of channel exclusion in general QRTs without convexity restriction. 
We consider a single-copy channel channel exclusion, where $m=1$ in the $m$-input channel discrimination.  
We fix the measurement strategy in this case; that is, when we quantify an advantage of a resource state, we use the same measurement for the resource state and the free states. 
Let $\{p_i,\Lambda_i\}_i$ be a channel ensemble, and $\{M_i\}_i$ be a given measurement. 
We define the error ratio of a resource state $\rho$ with respect to $\mathcal{F}_k(\mathcal{H})$ as 
\begin{equation}
    \label{eq:k-advantage exclusion}
    \dfrac{p_{\mathrm{err}}(\rho,\{p_i,\Lambda_i\}_i,\{M_i\}_i)}{\min_{\sigma_k \in \mathcal{F}_k(\mathcal{H})} p_{\mathrm{err}}(\sigma_k,\{p_i,\Lambda_i\}_i,\{M_i\}_i)}. 
\end{equation}
The worst-case error ratio of $\rho$ is the supremum of Eq.~\eqref{eq:k-advantage exclusion} over $k$. 
Since we would like to minimize the error probability in channel exclusion, we introduce the minimum error ratio as 
\begin{equation}
    \min_{\sigma_k \in \mathcal{F}_k(\mathcal{H})}\dfrac{p_{\mathrm{err}}(\rho,\{p_i,\Lambda_i\}_i,\{M_i\}_i)}{\min_{\sigma_k \in \mathcal{F}_k(\mathcal{H})} p_{\mathrm{err}}(\sigma_k,\{p_i,\Lambda_i\}_i,\{M_i\}_i)}. 
\end{equation}
The worst-case advantage of $\rho$ in this case is defined as 
\begin{equation}
    1 - \sup_{k} \min_{\sigma_k \in \mathcal{F}_k(\mathcal{H})} \dfrac{p_{\mathrm{err}}(\rho,\{p_i,\Lambda_i\}_i,\{M_i\}_i)}{\min_{\sigma_k \in \mathcal{F}_k(\mathcal{H})} p_{\mathrm{err}}(\sigma_k,\{p_i,\Lambda_i\}_i,\{M_i\}_i)}
\end{equation}

We find that the worst-case advantage for channel exclusion is quantified by the generalized robustness. 
\begin{theorem}
\label{thm:advantage worst case exclusion}
    For any resource state $\rho \in \mathcal{F}(\mathcal{H})\backslash \mathcal{D}(\mathcal{H})$, 
    \begin{equation}
        \begin{aligned}
            &1 - \sup_{k} \min_{\{p_i,\Lambda_i\}_i,\{M_i\}_i} \dfrac{p_{\mathrm{err}}(\rho,\{p_i,\Lambda_i\}_i,\{M_i\}_i)}{\min_{\sigma_k \in \mathcal{F}_k(\mathcal{H})} p_{\mathrm{err}}(\sigma_k,\{p_i,\Lambda_i\}_i,\{M_i\}_i)} \\  
            &= \mathrm{WoR}_{\mathcal{F}(\mathcal{H})}(\rho). 
        \end{aligned}
    \end{equation}
\end{theorem}
The proof is similar to that of Theorem~\ref{thm:advantage worst case}. 
We give the detail of the proof in Appendix~\ref{app:proof_worst_case_exclusion}. 

\section{Generalized robustness of quantum channels}
\label{sec:extension_channel}
In this section, we characterize an extension of our results to QRTs of quantum channels. 
Here, we will present the analysis of the generalized robustness for quantum channels, but the same extension should also be applicable to the weight of resource using the procedure of Sec.~\ref{sec:extension_weight} in place of those of Secs.~\ref{sec:multi-copy_witness},~\ref{sec:advantage_multi-copy}, and~\ref{sec:worst-case}. 

Let $\mathcal{H}_1$ and $\mathcal{H}_2$ be Hilbert spaces with dimensions $d_1$ and $d_2$, respectively. 
We consider a subset $\mathcal{O}_{\textup{F}}(\mathcal{H}_1\to \mathcal{H}_2)$ of quantum channels from $\mathcal{H}_1$ to $\mathcal{H}_2$, which we call the set of free channels. 
For brevity, we write $\mathcal{O}_{\textup{F}}$ when the Hilbert spaces are obvious from the context.  
The generalized robustness $R_{\mathcal{O}_{\textup{F}}}$ for quantum channels is defined in a similar way we defined the generalized robustness for quantum states~\cite{Takagi2019a}. 
    \begin{definition}
    The generalized robustness $R_{\mathcal{O}_{\textup{F}}}$ is defined as 
    \begin{equation}
        \begin{aligned}
        R_{\mathcal{O}_{\textup{F}}}(\Lambda) \coloneqq \min\Bigg\{s\geq 0 : 
        &\frac{\Lambda + s \Theta}{1+s} \in \mathcal{O}_{\textup{F}}(\mathcal{H}_1\to \mathcal{H}_2),\\ 
        &\Theta \in \mathcal{O}(\mathcal{H}_1\to \mathcal{H}_2) \Bigg\}
        \end{aligned}
    \end{equation}
    for all channels $\Lambda \in \mathcal{O}(\mathcal{H}_1\to \mathcal{H}_2)$. 
\end{definition}

\subsection{Multi-copy witness for QRT of channels}
Here, we show the construction of multi-copy witnesses for quantum channels, which act on multiple copies of Choi operators. 

To construct a multi-copy witness for quantum channels, we rewrite the generalized robustness using Choi operators. 
\begin{lemma}
    The generalized robustness can be expressed as 
    \begin{equation}
        \begin{aligned}
        \label{eq:robustness_choi}
        R_{\mathcal{O}_{\textup{F}}}(\Lambda) = \min\Bigg\{s\geq 0 : 
        &\frac{J_{\Lambda} + s J}{1+s} \in \mathcal{O}_{\textup{F}}^{J}(\mathcal{H}_1\to \mathcal{H}_2),\\ 
        &J \in \mathcal{O}^{J}(\mathcal{H}_1\to \mathcal{H}_2) \Bigg\}
        \end{aligned}
    \end{equation}
    for all quantum channels $\Lambda$, where $\mathcal{O}^{J}$ and $\mathcal{O}_{\textup{F}}^{J}$ denote the sets of Choi operators as in Eq.~\eqref{eq:O_J}
\end{lemma}
\begin{proof}
    Recall that there is a one-to-one correspondence between channels and Choi operators; that is, two channels $\Lambda_1$ and $\Lambda_2$ are the same if and only if $J_{\Lambda_1} = J_{\Lambda_2}$. 

    First, suppose that we have an optimal decomposition 
    \begin{equation}
        \Xi = \frac{\Lambda + s\Theta}{1+s}
    \end{equation}
    achieving $s = R_{\mathcal{O}_{\textup{F}}}(\Lambda)$, where $\Theta \in \mathcal{O}(\mathcal{H}_1 \to \mathcal{H}_2)$ and $\Xi \in \mathcal{O}_{\textup{F}}$. 
    By applying $\mathcal{I}_{\mathcal{H}_1}\otimes \Xi$ to $\sum_{i,j = 1}^{d_1} \ketbra{i}{j}\otimes\ketbra{i}{j}$, we have 
    \begin{equation}
        \begin{aligned}
            J_{\Xi} 
            &= \sum_{i,j = 1}^{d_1} \ketbra{i}{j}\otimes \left(\frac{\Lambda+s\Theta}{1+s}\right)(\ketbra{i}{j}) \\ 
            &= \sum_{i,j = 1}^{d_1} \ketbra{i}{j}\otimes \left(\frac{1}{1+s}\Lambda(\ketbra{i}{j}) +\frac{s}{1+s}\Theta(\ketbra{i}{j})\right) \\ 
            &= \frac{1}{1+s} \sum_{i,j = 1}^{d_1} \ketbra{i}{j}\otimes \Lambda(\ketbra{i}{j}) + \frac{s}{1+s}\sum_{i,j = 1}^{d_1} \ketbra{i}{j}\otimes \Theta(\ketbra{i}{j}) \\ 
            &= \frac{J_{\Lambda} + s J_{\Theta}}{1+s}. 
        \end{aligned}
    \end{equation}
    Thus, we have 
    \begin{equation}
        \begin{aligned}
        R_{\mathcal{O}_{\textup{F}}}(\Lambda) \leq \min\Bigg\{s\geq 0 : 
        &\frac{J_{\Lambda} + s J}{1+r} \in \mathcal{O}_{\textup{F}}^{J}(\mathcal{H}_1\to \mathcal{H}_2),\\ 
        &J \in \mathcal{O}^{J}(\mathcal{H}_1\to \mathcal{H}_2) \Bigg\}. 
        \end{aligned}
    \end{equation}
    Next, suppose that we have an optimal decomposition 
    \begin{equation}
        J^\prime = \frac{J_{\Lambda} + sJ}{1+s} 
    \end{equation}
    achieving the minimum in the right-hand side of Eq.~\eqref{eq:robustness_choi}, where $J \in \mathcal{O}^{J}(\mathcal{H}_1 \to \mathcal{H}_2)$ and $J^{\prime} \in \mathcal{O}^{J}_{\textup{F}}$. 
    There exist a channel $\Theta \in \mathcal{O}(\mathcal{H}_1 \to \mathcal{H}_2)$ and a free channel $\Xi \in \mathcal{O}_{\textup{F}}$ such that $J = J_{\Theta}$ and $J^{\prime} = J_{\Xi}$, respectively, and we have 
    \begin{equation}
        \begin{aligned}
        J_{\Xi}
        & = \frac{J_{\Lambda} + sJ_{\Theta}}{1+s}\\
        &= \frac{1}{1+s} \sum_{i,j = 1}^{d_1} \ketbra{i}{j}\otimes \Lambda(\ketbra{i}{j}) + \frac{s}{1+s}\sum_{i,j = 1}^{d_1} \ketbra{i}{j}\otimes \Theta(\ketbra{i}{j}) \\ 
        &=\sum_{i,j = 1}^{d_1} \ketbra{i}{j}\otimes \left(\frac{\Lambda+s\Theta}{1+s}\right)(\ketbra{i}{j}) \\ 
        &= J_{\tfrac{\Lambda + s\Theta}{1+s}}.  
        \end{aligned}
    \end{equation}
    Therefore, due to the one-to-one correspondence between channels and Choi operators, we have 
    \begin{equation}
        \Xi = \frac{\Lambda + s\Theta}{1+s}; 
    \end{equation}
    thus, we have 
    \begin{equation}
        \begin{aligned}
        R_{\mathcal{O}_{\textup{F}}}(\Lambda) \geq \min\Bigg\{s\geq 0 : 
        &\frac{J_{\Lambda} + s J}{1+s} \in \mathcal{O}_{\textup{F}}^{J}(\mathcal{H}_1\to \mathcal{H}_2),\\ 
        &J \in \mathcal{O}^{J}(\mathcal{H}_1\to \mathcal{H}_2) \Bigg\}, 
        \end{aligned}
    \end{equation}
    which completes the proof. 
\end{proof}
We consider a structure of $\mathcal{O}^{J}(\mathcal{H}_1\to\mathcal{H}_2)$. 
Let $J \in \mathcal{O}^{J}(\mathcal{H}_1\to\mathcal{H}_2)$ be a Choi operator of some quantum channel; that is, $J$ is a positive operator on $\mathcal{H}_1 \otimes \mathcal{H}_2$ and satisfies $\tr_{\mathcal{H}_2}[J] = I_{d_1}$~\cite{watrous_2018}. 
Write 
\begin{equation}
    J = \sum_{i=0}^{d_1^2-1}\sum_{j = 0}^{d_2^2-1} c_{i,j} \lambda^{(1)}_i \otimes \lambda^{(2)}_j
\end{equation}
using the $d_1 \times d_1$ generalized Gell-Mann matrices $\{\lambda^{(1)}_i\}_{i=0}^{d_1^2-1}$ and the $d_2 \times d_2$ generalized Gell-Mann matrices $\{\lambda^{(2)}_j\}_{j=0}^{d_2^2-1}$, where we define $\lambda^{(1)}_{0} \coloneqq I_{d_1}$ and $\lambda^{(2)}_{0} \coloneqq I_{d_2}$.
Indeed, the set $\{\lambda_i^{(1)}\otimes \lambda_j^{(2)}: i = 0,1,\ldots,d_1^2, j = 0,1,\ldots, d_2^2\}$ serves as an orthogonal basis for the vector space of Hermitian operators on $\mathcal{H}_1\otimes \mathcal{H}_2$ because 
the cardinality of the set $\{\lambda_i^{(1)}\otimes \lambda_j^{(2)}: i = 0,1,\ldots,d_1^2, j = 0,1,\ldots, d_2^2\}$ is $(d_1d_2)^2$ and 
\begin{equation}
    \tr[(\lambda_{i^\prime}^{(1)}\otimes \lambda_{j^\prime}^{(2)})^\dagger(\lambda_i^{(1)}\otimes \lambda_j^{(2)})] = 0
\end{equation}
for all $(i^{\prime},j^{\prime}) \neq (i,j)$. 
Since $J$ is a positive semidefinite operator, $c_{i,j}$ are real numbers. 
Now, using the relation $\tr_{\mathcal{H}_2}[J] = I_{d_1}$, we have 
\begin{equation}
    \tr_{\mathcal{H}_2}[J] = \sum_{i=0}^{d_1^2 -1} d_2c_{i,0} \lambda^{(1)}_i = I_{d_1}. 
\end{equation}
Hence, 
\begin{align}
    &c_{0,0} = \frac{1}{d_2}, \\ 
    &c_{i,0} = 0, \,\,i = 1,2,\ldots, d_1^2-1. 
\end{align}
Therefore, the form of $J$ is restricted to 
\begin{equation}
    J = \frac{I_{d_1}\otimes I_{d_2}}{d_2} + \sum_{i = 0}^{d_1^2-1}\sum_{j=1}^{d_2^2-1} c_{i,j} \lambda^{(1)}_i\otimes \lambda^{(2)}_j. 
\end{equation}
That is, $J$ can be characterized by using $1 + d_1^2(d_2^2-1)$ real numbers $c_{i,j}$. 

Now, we show the construction of our multi-copy witnesses for quantum channels. 
In the construction, we use the following lemma, which is a Choi-operator version of Lemma~\ref{lem:witness_characterization}. 
We show the detail of the proof in Appendix~\ref{app:proof_witness_characterization_channel}. 

\begin{lemma}~\label{lem:witness_characterization_channel}
    Let $\Lambda \in \mathcal{O}(\mathcal{H}_1\to \mathcal{H}_2)$ be a quantum channel, and let $s > 0$. 
    Consider a quantum channel $\Xi \in \mathcal{O}(\mathcal{H}_1\to \mathcal{H}_2)$.  
    Then, 
    $J_\Xi$ can be expressed as 
    \begin{equation}
        J_\Xi = \frac{J_\Lambda + s' J_\Theta}{1 + s'}
    \end{equation}
    by using some positive number $0 < s' \leq s$ and some channel $\Theta \in \mathcal{O}(\mathcal{H}_1\to \mathcal{H}_2)$ 
    if and only if 
    \begin{equation}
        S_m\left(\frac{1+s}{s} J_\Xi - \frac{1}{s}J_\Lambda \right) \geq 0
    \end{equation}
    for all $m = 1,2,\ldots,d_1d_2$. 
\end{lemma}

Fix a resource channel $\Lambda \in \mathcal{O}(\mathcal{H}_1\to \mathcal{H}_2) \backslash \mathcal{O}_{\textup{F}}(\mathcal{H}_1\to \mathcal{H}_2)$. 
For a channel $\Xi \in \mathcal{O}(\mathcal{H}_1\to \mathcal{H}_2)$, define 
\begin{equation}
    S_{m,J_\Lambda,s}(J_\Xi) \coloneqq S_m\left(\frac{1+s}{s} J_\Xi - \frac{1}{s}J_\Lambda \right) 
\end{equation}
for $m = 1,2,\ldots,d$.
Note that 
\begin{equation}
    S_{1,\rho,s}(J_\Xi) = \frac{1+s}{s}\tr[J_\Xi] - \frac{1}{s}\tr[J_\Lambda] = d_1\geq 0, 
\end{equation}
and thus we focus on $m\geq 2$.

It trivially follows that $S_{m,J_\Lambda,s}(J_\Lambda) = S_m(J_\Lambda) \geq 0$ for all $m$ since $J_\Lambda$ is a Choi operator, in particular, positive. 
When $s < R_{\mathcal{O}_{\textup{F}}}(\Lambda)$, by the definition of the generalized robustness, for any free channel $\Upsilon \in \mathcal{O}_{\textup{F}}(\mathcal{H}_1\to \mathcal{H}_2)$, there is no positive number $s' \leq  s$ such that 
\begin{equation}
    J_\Upsilon = \frac{J_\Lambda + s'J_\Theta}{1+s'}
\end{equation}
with some channel $\Theta \in \mathcal{O}(\mathcal{H}_1\to \mathcal{H}_2)$.
By Lemma~\ref{lem:witness_characterization_channel}, in this case, there exists $2 \leq m \leq d_1d_2$ such that $S_{m,J_\Lambda,s}(J_\Upsilon) < 0$
On the other hand, when $s \geq R_{\mathcal{O}_{\textup{F}}}(\Lambda)$, there exists $0 \leq s' \leq s$ such that 
\begin{equation}
    J_\Upsilon = \frac{J_\Lambda + s'J_\Theta}{1+s'}
\end{equation}
with some channel $\Theta \in \mathcal{O}(\mathcal{H}_1\to \mathcal{H}_2)$.
Indeed, we can take a free channel $\Upsilon \in \mathcal{O}_{\textup{F}}(\mathcal{H}_1\to \mathcal{H}_2)$ achieving $R$ with $s' = R$. 
In this case, for all $m = 2,3,\ldots, d_1d_2$, 
we have $S_{m,J_\Lambda,s}(J_\Upsilon)  \geq 0$ for such $\Upsilon$. 

Then, for all $2\leq m\leq d_1 d_2$ and $s>0$, we construct $W^{(\mathcal{O})}_m(\Lambda,s)$ satisfying 
\begin{equation}
    \tr\left[W^{(\mathcal{O})}_m(\Lambda,s) J_\Xi^{\otimes m}\right] = S_{m,\Lambda,s}(J_\Xi)
\end{equation}
for $\Xi \in \mathcal{O}(\mathcal{H}_1\to \mathcal{H}_2)$. 

Thus, with a similar construction in Theorem~\ref{cor:high_order_witness_qudit_conventional} by replacing $W_m$ with $W_m^{(\mathcal{O})}$, we have a Hermitian operator $W_{\mathcal{O}}(\Lambda,s)$ such that 
\begin{align}
        &\tr[W_{\mathcal{O}}(\Lambda,s) J_{\Lambda}^{\otimes d_1d_2}] < 0, \\
        &\tr[W_{\mathcal{O}}(\Lambda,s) J_{\Xi}^{\otimes d_1d_2}] \geq 0, \,\,\forall \Xi \in \mathcal{O}_{\textup{F}},
    \end{align}
if and only if $s < R_{\mathcal{O}_{\textup{F}}}(\Lambda)$. 

\subsection{Worst-case advantage of channel in state discrimination}
In this section, we prove that the generalized robustness for quantum channels is interpreted as a worst-case advantage in a \textit{state} discrimination task. 
We consider the following task. 
Let $\{p_i,\eta_i\}_i$ be an ensemble of quantum states, where $\eta_i \in \mathcal{D}(\mathcal{H}_1^{\otimes 2})$. 
We are given a state $\eta_i$ sampled from this ensemble with probability $p_i$. 
Our goal is to identify which state we are given with an application of $\mathcal{I}_{\mathcal{H}_1}\otimes \Lambda$, where $\Lambda \in \mathcal{O}(\mathcal{H}_1\to\mathcal{H}_2)$ is a given channel. 
For the identification, we choose a measurement $\{M_i\}_i$, whose label $i$ corresponds to the label of the state $\eta_i$. 
That is, if the measurement result is $i$, then we conclude that the state was $\eta_i$. 
Thus, the success probability of this state discrimination task with $\Lambda$, $\{p_i,\eta_i\}_i$, and $\{M_i\}_i$ is defined by 
\begin{equation}
    p_{\mathrm{succ}}(\Lambda, \{p_i,\eta_i\}_i, \{M_i\}_i) \coloneqq \sum_{i} p_i\tr[M_i (\mathcal{I}_{\mathcal{H}_1}\otimes \Lambda) (\eta_i)]. 
\end{equation}
Now, similar to the argument in Sec.~\ref{sec:worst-case}, we consider a situation where the set $\mathcal{O}_{\textup{F}}$ of free channels is given as a union 
\begin{equation}
    \mathcal{O}_{\textup{F}} \coloneqq \bigcup_k \mathcal{O}_{\textup{F}_k}, 
\end{equation}
where for all $k$, $\mathcal{O}_{\textup{F}_k}$ is a convex set. 
Let us consider the case where the ensemble $\{p_i,\eta_i\}_i$ and the measurement strategy $\{M_i\}_i$ are fixed. 
In this setup, the advantage of a resource channel $\Lambda$ in this state discrimination task with respect to the set $\mathcal{O}_{\textup{F}_k}$ is defined as 
\begin{equation}
    \label{eq:k_succ_prob_channel}
    \dfrac{p_{\mathrm{succ}}(\Lambda, \{p_i,\eta_i\}_i, \{M_i\}_i)}{\max_{\Xi_k \in \mathcal{O}_{\textup{F}_k}}p_{\mathrm{succ}}(\Xi_k, \{p_i,\eta_i\}_i, \{M_i\}_i)}. 
\end{equation}
The worst-case advantage is given as the infimum of Eq.~\eqref{eq:k_succ_prob_channel}. 
Similar to Theorem~\ref{thm:advantage worst case}, we find that the worst-case maximum advantage of state discrimination is quantified by the generalized robustness of quantum channels. 
The proof is given in Appendix~\ref{app:proof_worst_case_channel}. 
\begin{theorem}\label{thm:advantage worst case channel}
    For any resource channel $\Lambda$, 
    \begin{equation}
        \begin{aligned}
            &\inf_{k} \max_{\{p_i,\eta_i\}_i,\{M_i\}_i} \dfrac{p_{\mathrm{succ}}(\Lambda,\{p_i,\eta_i\}_i,\{M_i\}_i)}{\max_{\Xi_k \in \mathcal{O}_{\textup{F}_k}} p_{\mathrm{succ}}(\Xi_k,\{p_i,\eta_i\}_i,\{M_i\}_i)} \\  
            &= 1 + R_{\mathcal{O}_{\textup{F}}}(\Lambda). 
        \end{aligned}
    \end{equation}
\end{theorem}

\subsection{Generalized robustness of quantum instrument}
In this section, we show that a similar analysis also applies to the case of quantum instruments. 
Let $\mathcal{H}_1$ and $\mathcal{H}_2$ be Hilbert spaces with dimensions $d_1$ and $d_2$, respectively. 
Quantum instruments are defined in the following manner.  

\begin{definition}
    Let $m$ be a positive integer. 
    A family of completely positive maps $(\mathcal{E}_i \in \mathcal{O}^{(\mathrm{CP})}(\mathcal{H}_1\to \mathcal{H}_2): i = 0,1,\ldots,m-1)$ is called a \textit{quantum instrument} if 
    \begin{equation}
        \sum_{i = 1}^m \mathcal{E}_i \in \mathcal{O}(\mathcal{H}_1\to \mathcal{H}_2). 
    \end{equation}
    The set of quantum instruments from $\mathcal{H}_1$ to $\mathcal{H}_2$ with $m$ elements is denoted by $\mathcal{O}^{(m)}(\mathcal{H}_1\to \mathcal{H}_2)$.
\end{definition}

We introduce a framework of QRTs for quantum instruments. 
The general framework of dynamical resource theories that encompass the settings involving ensembles of quantum processes has recently been established~\cite{Regula2021one-shot}.
Along with this line of research, we consider a subset $\mathcal{O}^{(\mathrm{CP})}_\textup{F}(\mathcal{H}_1\to \mathcal{H}_2)$ of $\mathcal{O}^{(\mathrm{CP})}(\mathcal{H}_1\to \mathcal{H}_2)$ of instruments to be free. 
We introduce the generalized robustness $R_{\mathcal{O}^{(\mathrm{CP})}_\textup{F}}$ for quantum instruments. 
    \begin{definition}
    The generalized robustness $R_{\mathcal{O}^{(\mathrm{CP})}_\textup{F}}$ is defined as 
    \begin{equation}
        \begin{aligned}
        R_{\mathcal{O}_\textup{F}^{(\mathrm{CP})}}((\mathcal{E}_i)_i) \coloneqq \min\Bigg\{s\geq 0 : &\forall i,\,\frac{\mathcal{E}_i + s \mathcal{G}_i}{1+s} \in \mathcal{O}_{\textup{F}}^{(\mathrm{CP})}(\mathcal{H}_1\to \mathcal{H}_2),\\ 
        &(\mathcal{G}_i)_i \in \mathcal{O}^{(m)}(\mathcal{H}_1\to \mathcal{H}_2) \Bigg\}
        \end{aligned}
    \end{equation}
    for all quantum instruments $(\mathcal{E}_i)_i \in \mathcal{O}^{(m)}(\mathcal{H}_1\to \mathcal{H}_2)$. 
\end{definition}

In this framework, we characterize the generalized robustness of quantum instruments, using our results on the generalized robustness of quantum channels in the previous section. 
In particular, the following theorem shows that the results on the generalized robustness of quantum channels extend to quantum instruments by considering the channels corresponding to the instruments.
These results verify that the observation in Ref.~\cite{Regula2021one-shot}, about the equivalence between the generalized robustness explicitly defined for ensembles of quantum processes and that for corresponding quantum channels, generally holds for quantum instruments.

\begin{theorem}
    Let $(\mathcal{E}_i)_i \in \mathcal{O}^{(m)}(\mathcal{H}_1\to \mathcal{H}_2)$ be a quantum instrument. 
    Let $\mathcal{H}_3$ be an $m$-dimensional Hilbert space, and let $\{\ket{i}\}_{i=0}^{m-1}$ be an orthonormal basis of $\mathcal{H}_3$. 
    We define a quantum channel $\Lambda_{(\mathcal{E}_i)_i} \in \mathcal{O}(\mathcal{H}_1 \to \mathcal{H}_3\otimes\mathcal{H}_2)$ by 
    \begin{equation}
       \Lambda_{(\mathcal{E}_i)_i} (\rho) \coloneqq \sum_{i=0}^{m-1} \ketbra{i}{i}\otimes \mathcal{E}_i(\rho)
    \end{equation}
    for $\rho \in \mathcal{D}(\mathcal{H}_1)$. 
    Define a subset $\mathcal{O}_{\textup{F}} \subseteq \mathcal{O}(\mathcal{H}_1 \to \mathcal{H}_3\otimes\mathcal{H}_2)$ by 
    \begin{equation}
        \begin{aligned}
            \mathcal{O}_{\textup{F}} \coloneqq \Bigg\{\Lambda_{(\mathcal{E}_i)_i}: &\forall i,\, \mathcal{E}_{i}\in \mathcal{O}^{(\mathrm{CP})}_{\textup{F}}(\mathcal{H}_1\to \mathcal{H}_2),\,  \\
            &\sum_{i=0}^{m-1} \mathcal{E}_i \in \mathcal{O}(\mathcal{H}_1\to \mathcal{H}_2) \Bigg\}.
        \end{aligned}
    \end{equation}
    Then, we have 
    \begin{equation}
        R_{\mathcal{O}_\textup{F}^{(\mathrm{CP})}}((\mathcal{E}_i)_i) = R_{ \mathcal{O}_{\textup{F}}}(\Lambda_{(\mathcal{E}_i)_i}). 
    \end{equation}
\end{theorem}
\begin{proof}
    First, suppose that $(\mathcal{G}_i)_i \in \mathcal{O}^{(\mathrm{CP})}(\mathcal{H}_1\to \mathcal{H}_2)$ achieves $R_{\mathcal{O}_\textup{F}^{(\mathrm{CP})}}((\mathcal{E}_i)_i)$; that is, for all $i=0,1,\ldots,m-1$, 
    \begin{equation}
       \frac{\mathcal{E}_i + R_{\mathcal{O}_\textup{F}^{(\mathrm{CP})}}((\mathcal{E}_i)_i) \mathcal{G}_i}{1 + R_{\mathcal{O}_\textup{F}^{(\mathrm{CP})}}((\mathcal{E}_i)_i)} \in \mathcal{O}^{(\mathrm{CP})}_{\textup{F}}(\mathcal{H}_1\to \mathcal{H}_2). 
    \end{equation}
    Since $\sum_{i=0}^{m-1} \mathcal{E}_i$ and $\sum_{i=0}^{m-1} \mathcal{G}_i$ are both quantum channels, 
    \begin{equation}
        \sum_{i=0}^{m-1}\frac{\mathcal{E}_i + R_{\mathcal{O}_\textup{F}^{(\mathrm{CP})}}((\mathcal{E}_i)_i) \mathcal{G}_i}{1 + R_{\mathcal{O}_\textup{F}^{(\mathrm{CP})}}((\mathcal{E}_i)_i)} 
    \end{equation}
    is also a quantum channel, and hence 
    \begin{equation}
        \sum_{i=0}^{m-1} \ketbra{i}{i} \otimes \frac{\mathcal{E}_i + R_{\mathcal{O}_\textup{F}^{(\mathrm{CP})}}((\mathcal{E}_i)_i) \mathcal{G}_i}{1 + R_{\mathcal{O}_\textup{F}^{(\mathrm{CP})}}((\mathcal{E}_i)_i)} 
    \end{equation}
    is contained in $\mathcal{O}_{\textup{F}}$. 
    By linearity, we have
    \begin{equation}
        \frac{\Lambda_{(\mathcal{E}_i)_i} + R_{\mathcal{O}_\textup{F}^{(\mathrm{CP})}}((\mathcal{E}_i)_i) \Lambda_{(\mathcal{G}_i)_i}}{1 + R_{\mathcal{O}_\textup{F}^{(\mathrm{CP})}}((\mathcal{E}_i)_i)} \in \mathcal{O}_{\textup{F}}. 
    \end{equation}
    Therefore, by definition, 
    \begin{equation}
R_{\mathcal{O}_\textup{F}^{(\mathrm{CP})}}((\mathcal{E}_i)_i) \geq R_{ \mathcal{O}_{\textup{F}}}(\Lambda_{(\mathcal{E}_i)_i}). 
    \end{equation}
    
    On the other hand, suppose that $\Lambda \in \mathcal{O}(\mathcal{H}_1 \to \mathcal{H}_3 \otimes \mathcal{H}_2)$ achieves $R_{ \mathcal{O}_{\textup{F}}}(\Lambda_{(\mathcal{E}_i)_i})$; that is, 
    \begin{equation}
        \frac{\Lambda_{(\mathcal{E}_i)_i} + R_{\mathcal{O}_{\textup{F}}}(\Lambda_{(\mathcal{E}_i)_i}) \Lambda}{1 + R_{ \mathcal{O}_{\textup{F}}}(\Lambda_{(\mathcal{E}_i)_i})} \in \mathcal{O}_{\textup{F}}. 
    \end{equation}
    Considering the form of $\Lambda_{(\mathcal{E}_i)_i}$ and channels in $\mathcal{O}_{\textup{F}}$, $\Lambda$ can be written as $\Lambda_{(\mathcal{G}_i)_i}$ using some quantum instrument $(\mathcal{G}_i)_i \in \mathcal{O}^{(m)}(\mathcal{H}_1 \to \mathcal{H}_2)$. 
    Therefore, we have 
    \begin{equation}
    \sum_{i=0}^{m-1} \ketbra{i}{i} \otimes \frac{\mathcal{E}_i + R_{\mathcal{O}_{\textup{F}}}((\mathcal{E}_i)_i) \mathcal{G}_i}{1 + R_{\mathcal{O}_{\textup{F}}}((\mathcal{E}_i)_i)} \in \mathcal{O}_{\textup{F}}, 
    \end{equation}
    and thus for all $i = 0,1,\ldots,m-1$, 
    \begin{equation}
        \frac{\mathcal{E}_i + R_{\mathcal{O}_{\textup{F}}}((\mathcal{E}_i)_i) \mathcal{G}_i}{1 + R_{\mathcal{O}_{\textup{F}}}((\mathcal{E}_i)_i)} \in \mathcal{O}^{(\mathrm{CP})}_{\textup{F}}(\mathcal{H}_1\to \mathcal{H}_2). 
    \end{equation}
    Therefore, we have 
    \begin{equation}
R_{\mathcal{O}_\textup{F}^{(\mathrm{CP})}}((\mathcal{E}_i)_i) \leq R_{ \mathcal{O}_{\textup{F}}}(\Lambda_{(\mathcal{E}_i)_i}). 
    \end{equation}
    
\end{proof}

\section{Conclusion}
\label{sec:conclusion}
In this paper, to address the fundamental question of when general quantum resources are actually useful, we gave characterizations of the generalized robustness and the weight of resource in general QRTs without convexity restriction. 

First, complementing the work presented in the companion Letter \cite{PRL}, we established \textit{multi-copy} witnesses characterized by the generalized robustness. 
We constructed a family of Hermitian operators characterized by a single real-valued parameter.
We showed that the Hermitian operators in this family serve as the multi-copy witnesses of quantum resources if and only if the parameter characterizing the family is smaller than the generalized robustness.
We also found that the witness discerns more resource states from free states for a larger parameter. 
With these multi-copy witnesses, we designed a family of multi-input channel discrimination tasks, for which every resource state shows an advantage over free states. 
Next, inspired by a situation where multiple constraints form a non-convex set of free states, we considered the set of free states to be written as a union of several convex sets, each of which represents a different constraint. 
Under this description, we considered an advantage of resources over each convex set. 
We showed the generalized robustness can be interpreted as \textit{worst-case advantage} of a given resource for channel discrimination and found that every quantum resource shows an advantage over free states in this scenario as well. {\blue We illustrated our results with a particular example of non-convex resource theory related to coherence in multiple bases in a single qubit.}

We then observed that a similar analysis is also applicable to the weight of resource. 
We indeed gave another construction of multi-copy witnesses based on the weight of resource and showed that every quantum resource is useful in channel exclusion using these multi-copy witnesses. 
Moreover, we found that the weight of resource quantifies the worst-case advantage for channel exclusion. 
We further showed that our results can be extended to the \textit{generalized robustness for quantum channels} in dynamical resource theories. 
In particular, the concept of the multi-copy witness can be generalized to the Choi operators of quantum channels, and every resource channel is useful for a state discrimination task.  
We further showed that this analysis for quantum channels generalizes to quantum instruments as well. 

Thus, we demonstrated that {\em all} quantum states, channels, and instruments are useful in some operational tasks even in general QRTs without convexity restriction, by introducing new characterizations of the generalized robustness and the weight of resource.  
Our results offer a more general insight than the previous research on QRTs in that we did not impose any convexity assumption on all of our results. 
We believe this leads to a better understanding of quantum mechanical properties and their utilization in the framework of QRTs. 

We conclude this paper by stating several possible future directions for our results. 
While we constructed multi-copy witnesses characterized by the generalized robustness and the weight of resource, the number of copies needed to separate a resource state from free states depends on which free state we focus on. It would be nice to look for an $m$-copy witness for some fixed positive integer $m$ that can be used for all free states. 
It would also be nice to investigate if there exists a multi-copy witness where the number of copies is constant so that it could be applied to a QRT on an infinite-dimensional state space. 
Additionally, it would be interesting to see if such a multi-copy witness has a direct connection to an operational meaning of the generalized robustness or the weight of resource. 
We also wonder if there is another task for which the generalized robustness or the weight of resource has an operational meaning without considering decomposition into convex sets. A further generalization of our results to quantum betting tasks \cite{PRXQuantum.3.020366}, which interpolate between discrimination and exclusion tasks, could also be performed.
Moreover, we are also interested in extending our analysis to QRTs of sets of states, which were previously considered in the context of coherence~\cite{Brunner2021}. 
{\blue 
In addition, it would be interesting to see if our analysis can be extended to the framework of pseudo-resources~\cite{haug2023pseudorandom}, \textit{e.g.}, pseudo-entanglement~\cite{aaronson2023quantum,bouland2023publickey,arnonfriedman2023computational} and pseudo-magic~\cite{gu2023little}. 
In these setups, pseudo-random states~\cite{Zhengfeng2018} that do not contain plenty of resources but cannot be efficiently distinguished from highly resourceful states are considered as ``pseudo-resources''. 
It would be worthwhile to determine if the operational advantages of such pseudo-resources based on the generalized robustness or the weight of resource could also be shown. 
}
Finally, it may be interesting to assess whether one can design a task based on a multi-copy witness for Choi operators in which every resource channel shows some advantage. \\

\begin{acknowledgments}
We are grateful to Bartosz Regula and Pauli Jokinen for discussions and feedback. 
K.K. was supported by a Mike and Ophelia Lazaridis Fellowship, a Funai Overseas Scholarship, and a Perimeter Residency Doctoral Award. 
R.T. acknowledges the support of JSPS KAKENHI Grant Number JP23K19028, and JST, CREST Grant Number JPMJCR23I3, Japan. 
G.A. acknowledges support by the UK Research and Innovation (UKRI) under BBSRC Grant No.~BB/X004317/1 and EPSRC Grant No.~EP/X010929/1.
H.Y. acknowledges JST PRESTO Grant Number JPMJPR201A, JPMJPR23FC, and MEXT Quantum Leap Flagship Program (MEXT QLEAP) JPMXS0118069605, JPMXS0120351339\@. A part of this work was carried out at the workshop ``Quantum resources: from mathematical foundations to operational characterisation'' held in Singapore in December 2022. 
\end{acknowledgments}

\bibliography{robustness}
\bibliographystyle{apsrmp4-2}

\clearpage
\onecolumngrid
\appendix

\renewcommand{\thetheorem}{\Alph{section}.\arabic{theorem}}
\renewcommand{\theconjecture}{\Alph{section}.\arabic{conjecture}}

{\blue
\setcounter{theorem}{0}
\setcounter{conjecture}{0}
\section{Lower bound of generalized robustness inspired by duality}
\label{appendix:generalized_robustness_lower_bound}

In this section, we derive a lower bound of the generalized robustness in general QRTs without convexity assumption, inspired by the relation between the generalized robustness and a linear resource witness in convex QRTs. 
In particular, we show the following relation. 
\begin{proposition}
    Let $\mathcal{H}$ be a $d$-dimensional quantum system. 
    For any state $\rho \in \mathcal{D}(\mathcal{H})$, we have 
    \begin{equation}
        \label{eq:robustness_lower_bound_actual}
        R_{\mathcal{F}(\mathcal{H})}(\rho) \geq \max\Big\{\tr\left[W\rho\right] - 1 : W\geq 0,\,\, \tr\left[W\sigma\right] \leq 1, \forall \sigma \in \mathcal{F}(\mathcal{H})\Big\}.
    \end{equation}
\end{proposition}
This relation was previously shown for convex QRTs~\cite{Brandao2005,Regula2017,Napoli2016,Piani2016,Takagi2019a,Takagi2019b}, and actually, the equality in~\eqref{eq:robustness_lower_bound_actual} holds for convex QRTs. 
Here, we give a proof for general QRTs without convexity assumption. 

\begin{proof} 
    Consider the convex hull of the set of free states, denoted by $\mathrm{conv}\left(\mathcal{F}(\mathcal{H})\right)$. 
    Since $\mathcal{F}(\mathcal{H}) \subseteq \mathrm{conv}\left(\mathcal{F}(\mathcal{H})\right)$ by definition, we have 
    \begin{equation}
        \label{eq:robustness_original_conv}
        R_{\mathrm{conv}\left(\mathcal{F}(\mathcal{H})\right)}(\rho) \leq R_{\mathcal{F}(\mathcal{H})}(\rho). 
    \end{equation}
    Hence, to prove~\eqref{eq:robustness_lower_bound_actual}, we will show 
    \begin{equation}
        \label{eq:robustness_conv_witness}
        R_{\mathrm{conv}\left(\mathcal{F}(\mathcal{H})\right)}(\rho) = \max\Big\{\tr\left[W\rho\right] - 1 : W\geq 0,\,\, \tr\left[W\sigma\right] \leq 1, \forall \sigma \in \mathcal{F}(\mathcal{H})\Big\}. 
    \end{equation}
    Since $\mathrm{conv}\left(\mathcal{F}(\mathcal{H})\right)$ is convex, we can apply the analysis for convex QRTs in Refs.~\cite{Regula2017,Takagi2019b} to $\mathrm{conv}\left(\mathcal{F}(\mathcal{H})\right)$, and have 
    \begin{equation}
        \label{eq:conv_robustness_duality}
        R_{\mathrm{conv}\left(\mathcal{F}(\mathcal{H})\right)}(\rho) = \max\Big\{\tr\left[W\rho\right] - 1 : W\geq 0,\,\, \tr\left[W\sigma\right] \leq 1, \forall \sigma \in \mathrm{conv}\left(\mathcal{F}(\mathcal{H})\right)\Big\}. 
    \end{equation}
    Here, it trivially follows that 
    \begin{equation}
        \tr\left[W\sigma\right] \leq 1, \forall \sigma \in \mathrm{conv}\left(\mathcal{F}(\mathcal{H})\right) \Rightarrow \tr\left[W\sigma\right] \leq 1, \forall \sigma \in \mathcal{F}(\mathcal{H}). 
    \end{equation}
    On the other hand, when $\tr\left[W\sigma\right] \leq 1$ holds for any $\sigma \in \mathcal{F}(\mathcal{H})$, for any convex combination 
    \begin{equation}
        \tilde{\sigma} = \sum_{i} p_i \sigma_i, \,\sigma_i \in \mathcal{F}(\mathcal{H}), 
    \end{equation}
    we have 
    \begin{equation}
        \tr\left[W\tilde{\sigma} \right] = \sum_{i}p_i  \tr\left[W\sigma_i\right] \leq \sum_{i}p_i = 1. 
    \end{equation}
    Therefore, we also have 
    \begin{equation}
         \tr\left[W\sigma\right] \leq 1, \forall \sigma \in \mathcal{F}(\mathcal{H}) \Rightarrow \tr\left[W\sigma\right] \leq 1, \forall \sigma \in \mathrm{conv}\left(\mathcal{F}(\mathcal{H})\right). 
    \end{equation}
    Thus, we have the equivalence relation 
    \begin{equation}
          \tr\left[W\sigma\right] \leq 1, \forall \sigma \in \mathrm{conv}\left(\mathcal{F}(\mathcal{H})\right)\Leftrightarrow \tr\left[W\sigma\right] \leq 1, \forall \sigma \in \mathcal{F}(\mathcal{H}), 
    \end{equation}
    and thus, 
    \begin{equation}
        \label{eq:equivalence_witness_original_conv}
        \max\Big\{\tr\left[W\rho\right] - 1 : W\geq 0,\,\, \tr\left[W\sigma\right] \leq 1, \forall \sigma \in \mathrm{conv}\left(\mathcal{F}(\mathcal{H})\right)\Big\} = \max\Big\{\tr\left[W\rho\right] - 1 : W\geq 0,\,\, \tr\left[W\sigma\right] \leq 1, \forall \sigma \in \mathcal{F}(\mathcal{H})\Big\}
    \end{equation}
    From~\eqref{eq:conv_robustness_duality} and~\eqref{eq:equivalence_witness_original_conv}, we have~\eqref{eq:robustness_conv_witness} as desired. 
\end{proof}

\begin{remark}
    When $\mathcal{F}(\mathcal{H})$ is convex, we have $ R_{\mathrm{conv}\left(\mathcal{F}(\mathcal{H})\right)}(\rho) = R_{\mathcal{F}(\mathcal{H})}(\rho)$, and thus the equality in~\eqref{eq:robustness_lower_bound_actual} naturally holds. 
    
\end{remark}

The above proof mostly relies on the previous analysis of the generalized robustness in convex QRTs. 
For self-containedness, we also give an alternative proof for~\eqref{eq:robustness_lower_bound_actual} that is independent of the previous analysis. 
This proof is also largely inspired by the duality in convex optimization~\cite{boyd_vandenberghe_2004}, but here it should be kept in mind that the set $\mathcal{F}(\mathcal{H})$ might not be convex any longer. 
\begin{proof}[Another proof for~\eqref{eq:robustness_lower_bound_actual}]
    In this proof, without otherwise being noted, any operator is assumed to be Hermitian. 
    Recall that the generalized robustness of $\rho$ can be written as the following optimization problem. 
    \begin{align}
        \label{eq:original_opt_robustness_1}
            \mathrm{minimize} \quad &s \\
            \label{eq:original_opt_robustness_2}
            \mathrm{subject\,\,to} \quad  &\frac{\rho + s\tau}{1 + s} = \sigma, \\
            \label{eq:original_opt_robustness_3}
            & \tau \in \mathcal{D}(\mathcal{H}), \\
            \label{eq:original_opt_robustness_4}
            &\sigma \in \mathcal{F}(\mathcal{H}), \\
            \label{eq:original_opt_robustness_5}
            &s \geq 0.
    \end{align}
    By replacing $s\tau$ with a positive semidefinite operator $Z$, we can reformulate the optimization problem~\eqref{eq:original_opt_robustness_1}-\eqref{eq:original_opt_robustness_5} as 
    \begin{align}
        \label{eq:opt_robustness_1}
            \mathrm{minimize} \quad &\tr\left[Z\right] \\
            \label{eq:opt_robustness_2}
            \mathrm{subject\,\,to} \quad  &\rho + Z \in \mathcal{F}_c(\mathcal{H}), \\
            \label{eq:opt_robustness_3}
            & Z \geq 0, 
    \end{align}
    where $\mathcal{F}_c(\mathcal{H})$ is a cone generated by $\mathcal{F}(\mathcal{H})$, defined as 
    \begin{equation}
        \mathcal{F}_c(\mathcal{H}) \coloneqq \left\{\lambda \sigma : \lambda>0,\,\sigma \in \mathcal{F}(\mathcal{H})\right\}. 
    \end{equation}
Now, consider the following Lagrangian constructed from~\eqref{eq:opt_robustness_1}-\eqref{eq:opt_robustness_3}. 
\begin{equation}
    L(\rho,Z,V,W) \coloneqq \tr\left[Z\right] - \tr\left[V(\rho+Z)\right] - \tr\left[WZ\right], 
\end{equation}
where $V \in \mathcal{F}_c^*(\mathcal{H})$ and $W \geq 0$ with the dual cone 
\begin{equation}
    \mathcal{F}_c^*(\mathcal{H}) \coloneqq \left\{X \in \Herm\left(\mathcal{H}\right): \tr[XY] \geq 0,\,\, \forall Y\in \mathcal{F}_c(\mathcal{H}) \right\}. 
\end{equation}
Considering the reformulation of the generalized robustness given in~\eqref{eq:opt_robustness_1}-\eqref{eq:opt_robustness_3}, we have 
\begin{equation}
    \label{eq:lagrangian_robustness}
    R_{\mathcal{F}(\mathcal{H})}(\rho) = \inf_{Z\geq 0,\,\rho + Z \in \mathcal{F}_c(\mathcal{H})} \sup_{V\in\mathcal{F}_c^*(\mathcal{H}),\, W\geq 0} L(\rho,Z,V,W) 
\end{equation}
because under the condition that $\rho + Z \in \mathcal{F}_c(\mathcal{H})$ and $Z\geq 0$, we have 
\begin{equation}
    \sup_{V\in\mathcal{F}_c^*(\mathcal{H}),\, W\geq 0} L(\rho,Z,V,W) = \tr\left[Z\right]. 
\end{equation}
To lower-bound $R_{\mathcal{F}(\mathcal{H})}$, we introduce a slightly different optimization problem 
\begin{equation}
    \label{eq:lagrangian_primal}
    p \coloneqq \inf_{Z} \sup_{V\in\mathcal{F}_c^*(\mathcal{H}),\, W\geq 0} L(\rho,Z,V,W), 
\end{equation}
where the condition on $Z$ in the infimum of~\eqref{eq:lagrangian_robustness} is now dropped. 
By definition, we have 
\begin{equation}
    \label{eq:relation_robustness_primal}
    p \leq R_{\mathcal{F}(\mathcal{H})}(\rho). 
\end{equation}
To have the desired lower bound, we consider the dual problem 
\begin{equation}
    d \coloneqq  \sup_{V\in\mathcal{F}_c^*(\mathcal{H}),\, W\geq 0} \inf_{Z}L(\rho,Z,V,W). 
\end{equation}
By the max-min inequality~\cite{boyd_vandenberghe_2004}, we have 
\begin{equation}
    \label{eq:weal_duality}
    d \leq p, 
\end{equation}
which is usually referred to as the weak duality. 
To analyze the value of $d$, we further rewrite the Lagrangian $L(\rho,Z,V,W)$ as 
\begin{equation}
    L(\rho,Z,V,W) = -\tr\left[V\rho\right] + \tr\left[(I_d - V - W)Z\right].  
\end{equation}
Thus, 
\begin{equation}
    \inf_{Z} L(\rho,Z,V,W) 
    =\begin{cases}
        -\tr\left[V\rho\right] & I_d - V - W = 0,\\
        -\infty & \mathrm{otherwise}. 
    \end{cases}
\end{equation}
Therefore, we have 
\begin{equation}
    \begin{aligned}
    \label{eq:dual_intermediate}
    d 
    &=  \sup_{V\in\mathcal{F}_c^*(\mathcal{H}),\, W\geq 0} \inf_{Z}L(\rho,Z,V,W) \\ 
    &= \sup\Bigg\{-\tr\left[V\rho\right]: V\in\mathcal{F}_c^*(\mathcal{H}),\, W\geq 0,\, I_d - V - W = 0\Bigg\}\\
    &= \sup\Bigg\{-\tr\left[V\rho\right]:  W\geq 0,\,I_d - W\in\mathcal{F}_c^*(\mathcal{H}),\, \Bigg\} \\ 
    &= \sup\Bigg\{\tr\left[W\rho\right] - 1: W\geq 0,\,\tr\left[(I_d-W)Y\right] \geq 0,\,\forall Y \in \mathcal{F}_c(\mathcal{H}) \Bigg\}\\ 
    &=\sup\Bigg\{\tr\left[W\rho\right] - 1: W\geq 0,\,\tr\left[W\sigma\right] \leq 1,\,\forall \sigma \in \mathcal{F}(\mathcal{H}) \Bigg\}. 
    \end{aligned}
\end{equation}
We claim that the domain  
\begin{equation} 
    \label{eq:domain_dual}
    \Big\{W\geq 0:\tr\left[W\sigma\right] \leq 1,\,\forall \sigma \in \mathcal{F}(\mathcal{H}) \Big\}
\end{equation}
of the supremum is compact to replace $\sup$ with {$\max$} in~\eqref{eq:dual_intermediate}. 
Indeed, due to the continuity of the trace,~\eqref{eq:domain_dual} is closed. 
Moreover, by taking a full-rank free state $\sigma_{\mathrm{full}} \in \mathcal{F}(\mathcal{H})$ with minimum eigenvalue $\lambda_{\min} > 0$, for any $W$ in~\eqref{eq:domain_dual}, we have 
\begin{equation}
    1 \geq \tr\left[W\sigma_{\mathrm{full}} \right] \geq \lambda_{\min} \tr\left[W\right] = \lambda_{\min}\|W\|_1, 
\end{equation}
that is, 
\begin{equation}
    \|W\|_1 \leq \frac{1}{\lambda_{\min}}, 
\end{equation}
which implies that~\eqref{eq:domain_dual} is bounded. 
Therefore,~\eqref{eq:domain_dual} is closed and bounded, \textit{i.e.,} compact, and we have 
\begin{equation}
\label{eq:relation_lagrangian_dual}
    d = \max\Bigg\{\tr\left[W\rho\right] - 1: W\geq 0,\,\tr\left[W\sigma\right] \leq 1,\,\forall \sigma \in \mathcal{F}(\mathcal{H}) \Bigg\}.
\end{equation}

Combining~\eqref{eq:relation_robustness_primal},~\eqref{eq:weal_duality}, and~\eqref{eq:relation_lagrangian_dual}, we have the desired inequality~\eqref{eq:robustness_lower_bound_actual}. 
\end{proof}
}

\setcounter{theorem}{0}
\setcounter{conjecture}{0}
\section{Characterization of quantum states based on the generalized robustness}
\label{app:proof_witness_characterization}
In this section, we prove  Lemma~\ref{lem:witness_characterization}, which is a key ingredient of the multi-copy witness based on the generalized robustness shown in Theorem~\ref{cor:high_order_witness_qudit_conventional}.  
Here, we restate the lemma for readability. 

\begin{lemma}[Lemma~\ref{lem:witness_characterization} in main text]
    Let $\rho$ be a quantum state, and let $s > 0$. 
    Consider a quantum state $\eta$.  
    Then, 
    $\eta$ can be expressed as 
    \begin{equation}
        \label{eq:eta_characterization}
        \eta = \frac{\rho + s' \tau}{1 + s'}
    \end{equation}
    by using some positive number $0 < s' \leq s$ and some state $\tau \in \mathcal{D}(\mathcal{H})$ 
    if and only if it holds for all $m = 1,2,\ldots,d$ that
    \begin{equation}
        \label{eq:positivity_condition}
        S_m\left(\frac{1+s}{s} \eta - \frac{1}{s}\rho \right) \geq 0,
    \end{equation}
     with $S_m$ given in Lemma~\ref{lem:bloch_ball}. 
\end{lemma}

\begin{proof}
    By Lemma~\ref{lem:bloch_ball}, the condition in Eq.~\eqref{eq:positivity_condition} is equivalent to positivity of Hermitian operator
    \begin{equation}
        \zeta \coloneqq \frac{1+s}{s} \eta - \frac{1}{s}\rho. 
    \end{equation}
    
    First, suppose that $\eta$ cannot be expressed as 
    \begin{equation}
        \eta = \frac{\rho + s' \tau}{1 + s'}
    \end{equation}
    with $0 < s' \leq s$ and $\tau \in \mathcal{D}(\mathcal{H})$.
    By way of contradiction, suppose that 
    \begin{equation}
        \zeta = \frac{1+s}{s} \eta - \frac{1}{s}\rho 
    \end{equation}
    is positive. 
    Since 
    \begin{equation}
        \tr[\zeta] = \frac{1+s}{s} \tr[\eta] - \frac{1}{s}\tr[\rho] = \frac{1+s}{s} - \frac{1}{s} = 1,  
    \end{equation}
    $\zeta$ is a state. 
    In addition, we have 
    \begin{equation}
        \eta = \frac{\rho + s\zeta}{1 + s}. 
    \end{equation}
    However, this contradicts the fact that $\eta$ cannot be written in the form in Eq.~\eqref{eq:eta_characterization} using $s' \leq s$. 
    Thus, in this case, $\zeta$ is not positive, and thus there exists $2 \leq m\leq d$ such that 
    \begin{equation}
        S_m\left(\frac{1+s}{s} \eta - \frac{1}{s}\rho \right) < 0.
    \end{equation}
    
    On the other hand, suppose that we have 
    \begin{equation}
        \eta = \frac{\rho + s' \tau}{1 + s'}
    \end{equation}
    with $0 < s' \leq s$ and $\tau \in \mathcal{D}(\mathcal{H})$.
    Then,  
    \begin{align}
        \zeta 
        &= \frac{1+s}{s} \eta - \frac{1}{s}\rho \\ 
        &= \frac{1+s}{s}\left(\frac{\rho + s' \tau}{1 + s'}\right) - \frac{1}{s}\rho \\ 
        &= \left(\frac{1 + s}{1 + s'} - 1\right)\frac{1}{s} \rho + \frac{(1+s)s'}{s(1+s')}\tau. 
    \end{align}
     Since $0 < s'\leq s$, we have $(1 + s)/(1 + s') \geq 1$, so $\zeta$ is positive. 
     Therefore, 
     \begin{equation}
        S_m\left(\frac{1+s}{s} \eta - \frac{1}{s}\rho \right) \geq 0
    \end{equation}
    for all $m = 1,2,\ldots,d$. 
\end{proof}

In the statement of Lemma~\ref{lem:witness_characterization}, we consider a characterization 
\begin{equation}
    \eta = \frac{\rho + q\tau}{1 + q}
\end{equation}
of a state $\eta$ with positive number $q > 0$ and some state $\tau$, inspired by the definition of the generalized robustness. 
We discuss the condition for $q$ and $\tau$ to characterize $\eta$. 

\begin{proposition}
\label{prop:expression_qudit}
Let $\mathcal{H}$ be a $d$-dimensional Hilbert space. 
Let $\rho \in \mathcal{D}(\mathcal{H})$ be a quantum state with Bloch vector $(r_j)_j$, 
and let $\eta \in \mathcal{D}(\mathcal{H})$ be any quantum state with Bloch vector $(x_j)_j$ with $\eta \neq \rho$. 
Then, for a state $\tau \in \mathcal{D}(\mathcal{H})$ with Bloch vector $(t_j)_j$, there exists $q > 0$ such that  
\begin{equation}
    \label{eq:char_prop_qudit}
        \eta = \dfrac{\rho + q\tau}{1+q} 
    \end{equation}
if and only if $\tau$ is on the ray
\begin{equation}
    \label{eq:ray_qudit}
    l: \dfrac{t_j - r_j}{x_j - r_j} = u,\,\,\forall j\,\,\,\,(u > 1).
\end{equation}
\end{proposition}

\begin{proof}
First, suppose that the state $\tau$ satisfies 
\begin{equation}
    \label{eq:ray}
    \dfrac{t_j - r_j}{x_j - r_j} = u 
\end{equation}
for some $u>1$.
We may rewrite Eq.~\eqref{eq:ray} as 
\begin{align}
    x_j &= \frac{t_j + (u-1)r_j}{u} = \dfrac{r_j + \tfrac{1}{u -1}t_j}{1 + \tfrac{1}{u - 1}}. 
\end{align}
Therefore, we have  
\begin{equation}
    \label{eq:char}
    \eta = \dfrac{\rho + \tfrac{1}{u -1}\tau}{1 + \tfrac{1}{u - 1}} \eqqcolon \dfrac{\rho + q\tau}{1 + q}, 
\end{equation}
where we introduce $q \coloneqq \tfrac{1}{u-1}$. 
Therefore, $\sigma$ can be represented in the form of Eq.~\eqref{eq:char} if $\tau$ is on the ray~\eqref{eq:ray}.  
On the other hand, suppose that there exists $\tau$ such that 
\begin{equation}
        \eta = \dfrac{\rho + q\tau}{1+q}. 
    \end{equation}
Then, the Bloch vector $(t_j)_j$ of $\tau$ must satisfy 
\begin{align}
    \dfrac{t_j - r_j}{x_j - r_j} = \frac{1}{q} + 1; 
\end{align}
that is, $\tau$ must be on the ray~\eqref{eq:ray_qudit} with $u = \tfrac{1}{q} + 1$. 
\end{proof}

\begin{remark}
When $\eta = \rho$, by taking $\tau = \rho$, for any $q > 0$, we can write 
\begin{equation}
    \eta = \frac{\rho + q\tau}{1 + q}. 
\end{equation}
\end{remark}

\setcounter{theorem}{0}
\setcounter{conjecture}{0}
\section{Derivation of the operator used for constructing the multi-copy witness}
\label{app:witness_derivation}
In this section, we show the derivation of Eq.~\eqref{eq:witness_derivation} in detail. 
Indeed, with the construction 
\begin{equation}
    \begin{aligned}
        W^{(l,m)}_{n_1,n_2,\ldots,n_{d^2-1}} = \left(\dfrac{d}{2(d-1)}\right)^{\tfrac{l}{2}}c^{(l,m)}_{n_1,n_2,\ldots,n_{d^2-1}} \left(\lambda_1\right)^{\otimes n_1} \otimes\left(\lambda_2\right)^{\otimes n_2} \cdots \otimes \left(\lambda_{d^2-1}\right)^{\otimes n_{d^2-1}}\otimes I^{\otimes (m-l)}, 
    \end{aligned}
\end{equation}
recalling that the state $\eta$ is written as 
\begin{equation}
    \eta = \frac{1}{d}\left(I + \sqrt{\dfrac{d(d-1)}{2}} \sum_{j=1}^{d^2-1} x_j\lambda_j \right)
\end{equation}
with generalized Bloch vector $(x_j)_j$, 
we have 
\begin{equation}
    \begin{aligned}
        &\tr\left[W^{(l,m)}_{n_1,n_2,\ldots,n_{d^2-1}}\eta^{\otimes m}\right] \\
        &= \tr\left[\left(\left(\dfrac{d}{2(d-1)}\right)^{\tfrac{l}{2}}c^{(l,m)}_{n_1,n_2,\ldots,n_{d^2-1}} \left(\lambda_1\right)^{\otimes n_1} \otimes\left(\lambda_2\right)^{\otimes n_2} \cdots \otimes \left(\lambda_{d^2-1}\right)^{\otimes n_{d^2-1}}\otimes I^{\otimes (m-l)}\right) \cdot \left(\frac{1}{d}\left(I + \sqrt{\dfrac{d(d-1)}{2}} \sum_{j=1}^{d^2-1} x_j\lambda_j \right)\right)^{\otimes m}\right]\\ 
        &= \left(\dfrac{d}{2(d-1)}\right)^{\tfrac{l}{2}}c^{(l,m)}_{n_1,n_2,\ldots,n_{d^2-1}} \times \dfrac{1}{d^m}\left(\dfrac{d(d-1)}{2}\right)^{\tfrac{l}{2}} \tr\left[\left(x_1\lambda_1^2\right)^{\otimes n_1}\otimes \cdots \otimes \left(x_{d^2-1}\lambda_{d^2-1}^2\right)^{\otimes n_{d^2-1}}\otimes I^{\otimes (m-l)} \right] + \tr\left[\mathrm{Traceless\,\,terms}\right]\\
        &= \left(\dfrac{d}{2(d-1)}\right)^{\tfrac{l}{2}}c^{(l,m)}_{n_1,n_2,\ldots,n_{d^2-1}} \times \dfrac{1}{d^m}\left(\dfrac{d(d-1)}{2}\right)^{\tfrac{l}{2}} \times 2^l d^{m-l} \times \prod_{j=1}^{d^2-1} \left(x_j\right)^{n_j}\\\ 
        &= \left(c^{(l,m)}_{n_1,n_2,\ldots,n_{d^2-1}} \prod_{j=1}^{d^2-1} \left(x_j\right)^{n_j}\right). 
    \end{aligned} 
\end{equation}

\setcounter{theorem}{0}
\setcounter{conjecture}{0}
\section{Derivation of 2-copy witness for single-qubit states}
\label{app:witness_derivation_qubit}
In Sec.~\ref{sec:special_case_qubit}, we claim that for a resource state $\rho$, a Hermitian operator 
\begin{equation}
    \begin{aligned}
    W(\rho,s)
        &\coloneqq \left(1-\dfrac{r_1^2+r_2^2+r_3^2}{s^2}\right)I_2\otimes I_2 
    - \dfrac{(1+s)^2}{s^2}\left[\sigma_1\otimes\sigma_1 + \sigma_2\otimes\sigma_2 + \sigma_3\otimes\sigma_3\right]\\
    &+\dfrac{1+s}{s^2}[r_1\left(I_2\otimes \sigma_1 + \sigma_1\otimes I_2\right) 
    +r_2\left(I_2\otimes \sigma_2 + \sigma_2\otimes I_2\right) 
    +r_3\left(I_2\otimes \sigma_3 + \sigma_3\otimes I_2\right)], 
    \end{aligned}
\end{equation}
serves as a 2-copy witness for single-qubits states if (and only if) $s < R_{\mathcal{F}(\mathcal{H})}$. 
Here, we show the claim. Indeed, recalling that a single-qubit state $\eta$ with Bloch vector $(x_1,x_2,x_3)$ is written as 
\begin{equation}
    \eta = \dfrac{1}{2}\left(I + x_1\sigma_1 + x_2\sigma_2 + x_3\sigma_3\right), 
\end{equation}
we have 
    \begin{equation}
        \label{eq:derivation_qubit}
        \begin{aligned}
        &\tr\left[W(\rho,s)\eta^{\otimes2}\right] \\  
        &= \tr\Bigg[\Bigg( \left(1 - \dfrac{r_1^2+r_2^2+r_3^2}{s^2}\right)I_2\otimes I_2 - \dfrac{(1+s)^2}{s^2}\left[\sigma_1\otimes\sigma_1 + \sigma_2\otimes\sigma_2 + \sigma_3\otimes\sigma_3\right]\\
        &+\dfrac{1+s}{s^2}\left[r_1\left(I_2\otimes \sigma_1 + \sigma_1\otimes I_2\right) 
        +r_2\left(I_2\otimes \sigma_2 + \sigma_2\otimes I_2\right) 
        +r_3\left(I_2\otimes \sigma_3 + \sigma_3\otimes I_2\right)\right]\Bigg)\cdot\left(\dfrac{1}{2}\left(I_2 + x_1\sigma_1 + x_2\sigma_2 + x_3\sigma_3\right)\right)^{\otimes 2}\Bigg] \\ 
        &= \frac{1}{4}\tr\Bigg[\Bigg( \left(1-\dfrac{r_1^2+r_2^2+r_3^2}{s^2}\right)I_2\otimes I_2 - \frac{(1+s)^2}{s^2}\sigma_1\otimes\sigma_1 - \frac{(1+s)^2}{s^2} \sigma_2\otimes\sigma_2 -  \frac{(1+s)^2}{s^2}\sigma_3\otimes\sigma_3 \\
        &+\dfrac{1+s}{s^2}r_1\left(I_2\otimes \sigma_1 + \sigma_1\otimes I_2\right) 
        +\dfrac{1+s}{s^2}r_2\left(I_2\otimes \sigma_2 + \sigma_2\otimes I_2\right) +\dfrac{1+s}{s^2}\left(I_2\otimes \sigma_3 + \sigma_3\otimes I_2\right) \Bigg) \\ 
        &\cdot \Bigg(I_2\otimes I_2 + x_1(I_2\otimes \sigma_1 +\sigma_1\otimes I_2) + x_2(I_2\otimes \sigma_2 +\sigma_2\otimes I_2) + x_3(I_2\otimes \sigma_3 +\sigma_3\otimes I_2) + x_1^2\sigma_1\otimes \sigma_1 + x_2^2\sigma_2\otimes \sigma_2 + x_3^2\sigma_3\otimes \sigma_3\Bigg) \Bigg] \\  
        &=\frac{1}{4}\tr\Bigg[\Bigg(1 - \dfrac{r_1^2+r_2^2+r_3^2}{s^2} + \dfrac{2(1+s)}{s^2}r_1x_1 +\dfrac{2(1+s)}{s^2}r_2x_2 +\dfrac{2(1+s)}{s^2}r_3x_3 -
        \frac{(1+s)^2}{s^2}x_1^2 - \frac{(1+s)^2}{s^2}x_2^2 - \frac{(1+s)^2}{s^2}x_3^2\Bigg)I_2\otimes I_2\Bigg] \\ 
        &\quad + \frac{1}{4}\tr\Bigg[\mathrm{Terms\,\,with\,\,}I_2\otimes \sigma_1, \sigma_1\otimes I_2,  
        I_2\otimes \sigma_2, \sigma_2\otimes I_2,  I_2\otimes \sigma_3, \sigma_3\otimes I_2, \\ 
        &\quad\quad\quad\quad\quad\quad 
        \sigma_1\otimes \sigma_2,  \sigma_2\otimes \sigma_1,  \sigma_1\otimes \sigma_3, 
        \sigma_3\otimes \sigma_1,  \sigma_2\otimes \sigma_3,  \sigma_3\otimes \sigma_2, 
        \sigma_1\otimes \sigma_1,  \sigma_2\otimes \sigma_2,  \sigma_3\otimes \sigma_3\Bigg] \\ 
        &=\frac{1}{4}\tr\Bigg[\dfrac{s^2 - (r_1^2+r_2^2+r_3^2)  +2(1+s)(r_1x_1 + r_2x_2 + r_3x_3) - (1+s)^2(x_1^2 + x_2^2 + x_3^2)}{s^2}I_2\otimes I_2\Bigg]\\
        &= S_2\left(\frac{1+s}{s}\eta - \frac{1}{s}\rho \right) \\ 
        &= S_{2,\rho,s}(\eta). 
        \end{aligned}
    \end{equation} 

\setcounter{theorem}{0}
\setcounter{conjecture}{0}
\section{Characterization of quantum states based on the weight of resource}
\label{app:proof_witness_characterization_weight}
In this section, we prove  Lemma~\ref{lem:witness_characterization_weight}.
The proof follows in a similar way to Lemma~\ref{lem:witness_characterization}, but we show the proof for completeness. 
Here, we restate the lemma. 

\begin{lemma}[Lemma~\ref{lem:witness_characterization_weight} in main text]
    Let $\rho$ be a quantum state, and let $0 < s < 1$. 
    Consider a quantum state $\eta$.  
    Then, 
    $\eta$ can be expressed as 
    \begin{equation}
        \label{eq:eta_characterization_weight}
        \eta = \frac{\rho - s' \tau}{1 - s'}
    \end{equation}
    by using some positive number $0 < s' \leq s$ and some state $\tau \in \mathcal{D}(\mathcal{H})$ 
    if and only if 
    \begin{equation}
        \label{eq:positivity_condition_weight}
        S_m\left(\frac{1}{s}\rho - \frac{1-s}{s} \eta\right) \geq 0
    \end{equation}
    for all $m = 1,2,\ldots,d$. 
\end{lemma}

\begin{proof}
    By Lemma~\ref{lem:bloch_ball}, the condition in Eq.~\eqref{eq:positivity_condition_weight} is equivalent to positivity of Hermitian operator
    \begin{equation}
        \zeta \coloneqq \frac{1}{s}\rho- \frac{1+s}{s} \eta. 
    \end{equation}
    
    First, suppose that $\eta$ cannot be expressed as 
    \begin{equation}
        \eta = \frac{\rho - s' \tau}{1 - s'}
    \end{equation}
    with $0 < s' \leq s < 1$ and $\tau \in \mathcal{D}(\mathcal{H})$.
    By way of contradiction, suppose that 
    \begin{equation}
        \zeta =\frac{1}{s}\rho - \frac{1-s}{s} \eta 
    \end{equation}
    is positive. 
    Since 
    \begin{equation}
        \tr[\zeta] =  \frac{1}{s}\tr[\rho] -\frac{1-s}{s} \tr[\eta] = \frac{1}{s}-\frac{1-s}{s} = 1,  
    \end{equation}
    $\zeta$ is a state. 
    In addition, we have 
    \begin{equation}
        \eta = \frac{\rho - s\zeta}{1 - s}. 
    \end{equation}
    However, this contradicts the fact that $\eta$ cannot be written in the form in Eq.~\eqref{eq:eta_characterization_weight} using $s' \leq s$. 
    Thus, in this case, $\zeta$ is not positive, and thus there exists $2 \leq m\leq d$ such that 
    \begin{equation}
        S_m\left(\frac{1}{s}\rho - \frac{1+s}{s} \eta \right) < 0
    \end{equation}
    
    On the other hand, suppose that we have 
    \begin{equation}
        \eta = \frac{\rho - s' \tau}{1 - s'}
    \end{equation}
    with $0 < s' \leq s$ and $\tau \in \mathcal{D}(\mathcal{H})$.
    Then,  
    \begin{align}
        \zeta 
        &=  \frac{1}{s}\rho - \frac{1-s}{s} \eta\\ 
        &= \frac{1}{s}\rho - \frac{1-s}{s}\left(\frac{\rho - s' \phi}{1 - s'}\right)\\ 
        &= \left(1 - \frac{1 - s}{1 - s'}\right)\frac{1}{s} \rho + \frac{(1-s)s'}{s(1-s')}\phi. 
    \end{align}
     Since $0 < s'\leq < 1$, we have $1 - (1 - s)/(1 - s') \geq 1$, so $\zeta$ is positive. 
     Therefore, 
     \begin{equation}
        S_m\left(\frac{1}{s}\rho - \frac{1-s}{s} \eta \right) \geq 0
    \end{equation}
    for all $m = 1,2,\ldots,d$. 
\end{proof}

In the statement of Lemma~\ref{lem:witness_characterization_weight}, we consider a characterization 
\begin{equation}
    \eta = \frac{\rho - q\tau}{1 - q}
\end{equation}
of a state $\eta$ with positive number $0<q<1$ and some state $\tau$, inspired by the definition of the weight of resource. 
We discuss the condition for $q$ and $\tau$ to characterize $\eta$. 

\begin{proposition}
\label{prop:expression_qudit_weight}
Let $\mathcal{H}$ be a $d$-dimensional Hilbert space. 
Let $\rho \in \mathcal{D}(\mathcal{H})$ be a given quantum state with Bloch vector $(r_j)_j$, 
and let $\eta  \in \mathcal{D}(\mathcal{H})$ be any quantum state with Bloch vector $(x_j)_j$ such that $\eta \neq \rho$. 
Then, for a state $\tau$ with Bloch vector $(t_j)_j$, there exists $0<q<1$ such that  
\begin{equation}
    \label{eq:char_prop_qudit_weight}
        \rho  = q\tau + (1-q)\eta
    \end{equation}
if and only if $\tau$ is on the ray
\begin{equation}
    \label{eq:ray_qudit_weight}
    l: \dfrac{t_j - r_j}{x_j - r_j} = u,\,\,\forall j\,\,\,\,(u < 0).
\end{equation}
\end{proposition}

\begin{proof}
The proof follows similarly to the proof of Proposition~\ref{prop:expression_qudit}. 
First, suppose that the state $\tau$ satisfies 
\begin{equation}
    \label{eq:ray_weight}
    l: \dfrac{t_j - r_j}{x_j - r_j} = u,\,\,\forall j
\end{equation}
for some $u < 0$. 
We may rewrite Eq.~\eqref{eq:ray_weight} as 
\begin{align}
    \label{eq:ray_weight_rewritten}
    r_j = \frac{1}{1-u}t_j + \left(1 - \frac{1}{1-u}\right)x_j. 
\end{align}
Therefore, we have 
\begin{equation}
    \label{eq:char_weight}
    \rho = \frac{1}{1-u}\tau + \left(1 - \frac{1}{1-u}\right)\eta = q\tau + (1-q)\eta,
\end{equation}
where we introduce $q \coloneqq \tfrac{1}{1-u}$. 
Note that since $u < 0$, we have $0 < q < 1$. 
Therefore, $\eta \neq \rho$ can be represented in the form of Eq.~\eqref{eq:char_weight} if $\tau$ is on the ray~\eqref{eq:ray_qudit_weight}.  
On the other hand, suppose that there exists $\tau$ such that 
\begin{equation}
        \rho = q\tau + (1-q)\eta
\end{equation}
with $0 < q < 1$. 
Then, the Bloch vector $(t_j)_j$ of $\tau$ must satisfy 
\begin{align}
    \dfrac{t_j - r_j}{x_j - r_j} = 1 - \frac{1}{q}; 
\end{align}
that is, $\tau$ must be on the ray~\eqref{eq:ray_qudit_weight} with $u = 1 - \tfrac{1}{q}$. 
Note that since $0 < q < 1$, we have $u < 0$. 
\end{proof}

\begin{remark}
When $\eta = \rho$, by taking $\tau = \rho$, for any $0 < q < 1$, we can write 
\begin{equation}
   \rho = q\tau + (1-q)\eta
\end{equation}
\end{remark}

\setcounter{theorem}{0}
\setcounter{conjecture}{0}
\section{Proof of operational advantage in multi-input channel exclusion task}
\label{app:proof_advantage_exclusion}
In this section, we prove Theorem~\ref{thm:advantage_exclusion}.
We restate the theorem for readability. 

\begin{theorem}[Theorem~\ref{thm:advantage_exclusion} in main text]
Let $\mathcal{H}$ be a $d$-dimensional Hilbert space. 
For any resource state $\rho \in \mathcal{D}(\mathcal{H})\backslash\mathcal{F}(\mathcal{H})$, there exists a family of channel ensembles $\left(\{p_i,\Lambda^{(m)}_i\}_i\right)_{m=2}^{d}$ such that 
\begin{equation}
    \begin{aligned}
        \max_{\sigma \in \mathcal{F}(\mathcal{H})} \min_{m = 2,3,\ldots,d} \frac{\displaystyle \min_{\{M_i\}_i} p_{\mathrm{err}}(\{p_i,\Lambda^{(m)}_i\}_i,\{M_i\}_i,\rho^{\otimes m})}{\displaystyle \min_{\{M_i\}_i} p_{\mathrm{err}}(\{p_i,\Lambda^{(m)}_i\}_i,\{M_i\}_i,\sigma^{\otimes m})}
        < 1. 
    \end{aligned} 
\end{equation}
\end{theorem}

\begin{proof}
Let $\{W^{(\mathrm{WoR})}_m\}_{m=2}^d$ be a family of Hermitian operators such that every free state $\sigma\in\mathcal{F}(\mathcal{H})$ comes with an integer $m\in\{2,\dots, d\}$ satisfying
\begin{align}
        &\tr[W^{(\mathrm{WoR})}_m\rho^{\otimes m}] > 1, \\ 
        &0\leq\tr[W^{(\mathrm{WoR})}_m\sigma^{\otimes m}] \leq 1.  
        \label{eq:modified witness free weight}
    \end{align}
    The existence of such a family of Hermitian operators is guaranteed by Theorem~\ref{thm:high_order_witness_qudit_weight}.
Indeed, it ensures that when $s < R_{\mathcal{F}(\mathcal{H})}(\rho)$, for every $\sigma\in\mathcal{F}(\mathcal{H})$ there exists $m\in\{2,\ldots, d\}$ satisfying
    \begin{align}
        \label{eq:witness_resource_weight}
        &\tr[\widetilde{W}^{(\mathrm{WoR})}_m(\rho,s)\rho^{\otimes m}] < 0, \\ 
        &\tr[\widetilde{W}^{(\mathrm{WoR})}_m(\rho,s)\sigma^{\otimes m}] \geq 0 
    \end{align}
    for the family $\widetilde{\mathcal{W}}^{\mathrm{WoR}}_{\rho,s} = (\widetilde{W}^{(\mathrm{WoR})}_m(\rho,s) \in \Herm(\mathcal{H}^{\otimes m}): m = 2,3,\ldots,d)$ of Hermitian operators.
    Then, we may take $W^{(\mathrm{WoR})}_m \coloneqq I_d^{\otimes m} - \widetilde{W}^{(\mathrm{WoR})}_m(\rho,s)/\|\widetilde{W}^{(\mathrm{WoR})}_m(\rho,s)\|_\infty$ for some $s < R_{\mathcal{F}(\mathcal{H})}(\rho)$. 
    Note that Eq.~\eqref{eq:witness_resource_weight} implies that $\widetilde{W}^{(\mathrm{WoR})}_m(\rho,s)$ has at least one negative eigenvalues, so we have $\|W^{(\mathrm{WoR})}_m\|_\infty > 1$. 
    
    Define two quantum channels $\Lambda^{(m)}_1$ and $\Lambda^{(m)}_2$ as 
        \begin{align}
            \Lambda^{(m)}_1(X) &= \left(\dfrac{\tr(X)}{2} + \dfrac{\tr(W^{(\mathrm{WoR})}_mX)}{2\|W^{(\mathrm{WoR})}_m\|_\infty}\right)\ketbra{0}{0} + \left(\dfrac{\tr(X)}{2} - \dfrac{\tr(W^{(\mathrm{WoR})}_mX)}{2\|W^{(\mathrm{WoR})}_m\|_\infty}\right)\ketbra{1}{1}, \\
            \Lambda^{(m)}_2(X) &= \left(\dfrac{\tr(X)}{2} - \dfrac{\tr(W^{(\mathrm{WoR})}_mX)}{2\|W^{(\mathrm{WoR})}_m\|_\infty}\right)\ketbra{0}{0} + \left(\dfrac{\tr(X)}{2} + \dfrac{\tr(W^{(\mathrm{WoR})}_mX)}{2\|W^{(\mathrm{WoR})}_m\|_\infty}\right)\ketbra{1}{1}. 
        \end{align}
    Note that these channels have the same form as the channels used in the proof of Theorem~\ref{thm:advantage_multi-copy}. 
    Let $\sigma\in\mathcal{F}(\mathcal{H})$ be an arbitrary free state and let $m$ be an integer that satisfies \eqref{eq:modified witness free weight}.
    Then, since $\tr[W^{(\mathrm{WoR})}_m\rho^{\otimes m}] > 1$, 
    \begin{equation}
        \begin{aligned}
         \|\Lambda^{(m)}_1(\rho^{\otimes m}) - \Lambda^{(m)}_2(\rho^{\otimes m})\|_1 
            = \left\|\dfrac{\tr[W^{(\mathrm{WoR})}_m\rho^{\otimes d}]}{\|W^{(\mathrm{WoR})}_m\|_\infty}\ketbra{0}{0} - \dfrac{\tr[W^{(\mathrm{WoR})}_m\rho^{\otimes m}]}{\|W^{(\mathrm{WoR})}_m\|_\infty}\ketbra{1}{1}\right\|_1 
            = \dfrac{2\tr[W^{(\mathrm{WoR})}_m\rho^{\otimes m}]}{\|W^{(\mathrm{WoR})}_m\|_\infty}  
            > \dfrac{2}{\|W^{(\mathrm{WoR})}_m\|_\infty}. 
        \end{aligned}
    \end{equation}
    Similarly, since $0\leq\tr[W^{(\mathrm{WoR})}_m\sigma^{\otimes m}] \leq 1$, 
    \begin{equation}
        \begin{aligned}
        \|\Lambda^{(m)}_1(\sigma^{\otimes m}) - \Lambda^{(m)}_2(\sigma^{\otimes m})\|_1 = \dfrac{2\tr[W^{(\mathrm{WoR})}_m\sigma^{\otimes m}]}{\|W^{(\mathrm{WoR})}_m\|_\infty} \leq \dfrac{2}{\|W^{(\mathrm{WoR})}_m\|_\infty}. 
        \end{aligned}
    \end{equation}
    Consider the channel ensemble consisting of $\Lambda^{(m)}_1$ and $\Lambda^{(m)}_2$ with probability $\tfrac{1}{2}$, denoted by $\{\tfrac{1}{2},\Lambda^{(m)}_i\}_{i=1}^2$. 
    Then, we have 
        \begin{equation}
            \begin{aligned}
            \min_{\{M_i\}_i} p_{\mathrm{err}}(\{\tfrac{1}{2},\Lambda^{(m)}_i\}_{i=1}^2,\{M_i\}_i,\rho^{\otimes m}) 
            = \dfrac{1}{2}\left(1 - \dfrac{1}{2}\|\Lambda^{(m)}_1(\rho^{\otimes m}) - \Lambda^{(m)}_2(\rho^{\otimes m})\|_1\right) 
            < \dfrac{1}{2}\left(1 - \dfrac{1}{\|W^{(\mathrm{WoR})}_m\|_\infty}\right), 
            \end{aligned}
        \end{equation}
        and 
        \begin{equation}
            \begin{aligned}
             \min_{\{M_i\}_i} p_{\mathrm{err}}(\{\tfrac{1}{2},\Lambda^{(m)}_i\}_{i=1}^2,\{M_i\}_i,\sigma^{\otimes m}) = \dfrac{1}{2}\left(1 - \dfrac{1}{2}\max_{\sigma\in\mathcal{F}(\mathcal{H})}\|\Lambda^{(m)}_1(\sigma^{\otimes m}) - \Lambda^{(m)}_2(\sigma^{\otimes m})\|_1\right) \geq \dfrac{1}{2}\left(1 - \dfrac{1}{\|W^{(\mathrm{WoR})}_m\|_\infty}\right). 
            \end{aligned}
        \end{equation}
    Note that since $\|W^{(\mathrm{WoR})}_m\|_{\infty} > 1$, we have $ 0 < 1 - 1/\|W^{(\mathrm{WoR})}_m\|_\infty < 1$, so 
    \begin{equation}
        0 < \dfrac{1}{2}\left(1 - \dfrac{1}{\|W^{(\mathrm{WoR})}_m\|_\infty}\right) < \frac{1}{2}. 
    \end{equation}
    Therefore, we have 
    \begin{equation}
        \begin{aligned}
            &\dfrac{\min_{\{M_i\}_i} p_{\mathrm{err}}(\{\tfrac{1}{2},\Lambda^{(m)}_i\}_{i=1}^2,\{M_i\}_i,\rho^{\otimes m})}{\min_{\{M_i\}_i} p_{\mathrm{err}}(\{\tfrac{1}{2},\Lambda^{(m)}_i\}_{i=1}^2,\{M_i\}_i,\sigma^{\otimes m})} 
            < 1. 
            \label{eq:advantage fixed free state weight}
        \end{aligned}
    \end{equation}

    Since every free state comes with an integer $m$ satisfying \eqref{eq:advantage fixed free state weight}, we have
    \begin{equation}
    \begin{aligned}
        \min_{m = 2,3,\ldots,d} \frac{\displaystyle \min_{\{M_i\}_i} p_{\mathrm{err}}(\{p_i,\Lambda^{(m)}_i\}_i,\{M_i\}_i,\rho^{\otimes m})}{\displaystyle \min_{\{M_i\}_i} p_{\mathrm{err}}(\{p_i,\Lambda^{(m)}_i\}_i,\{M_i\}_i,\sigma^{\otimes m})}
        < 1
    \end{aligned} 
\end{equation}
for arbitrary $\sigma\in\mathcal{F}(\mathcal{H})$.
    Taking maximum over all free states on the left-hand side completes the proof. 
\end{proof}

\setcounter{theorem}{0}
\setcounter{conjecture}{0}
\section{Proof of worst-case advantage in channel exclusion task}
\label{app:proof_worst_case_exclusion}
In this section, we prove Theorem~\ref{thm:advantage worst case exclusion}.
We restate the theorem for readability. 

\begin{theorem}[Theorem~\ref{thm:advantage worst case exclusion} in main text]
    For any resource state $\rho \in \mathcal{F}(\mathcal{H})\backslash \mathcal{D}(\mathcal{H})$, 
    \begin{equation}
            1 - \sup_{k} \min_{\{p_i,\Lambda_i\}_i,\{M_i\}_i} \dfrac{p_{\mathrm{err}}(\rho,\{p_i,\Lambda_i\}_i,\{M_i\}_i)}{\min_{\sigma_k \in \mathcal{F}_k(\mathcal{H})} p_{\mathrm{err}}(\sigma_k,\{p_i,\Lambda_i\}_i,\{M_i\}_i)} = \mathrm{WoR}_{\mathcal{F}(\mathcal{H})}(\rho). 
    \end{equation}
\end{theorem}

\begin{proof}
Using the same argument in the proof of Lemma~\ref{lem:worst_case_robustness}, for any state $\rho$, we have 
    \begin{equation}
        \mathrm{WoR}_{\mathcal{F}(\mathcal{H})}(\rho) = \inf_{k} \mathrm{WoR}_{\mathcal{F}_k(\mathcal{H})}(\rho). 
    \end{equation}
Now, for any $k$, since $\mathcal{F}_k(\mathcal{H})$ is convex, we can apply the analysis from Ref.~\cite{Ducuara2020} and get 
\begin{equation}
    \begin{aligned}
        1 - \min_{\{p_i,\Lambda_i\}_i,\{M_i\}_i} \dfrac{p_{\mathrm{err}}(\rho,\{p_i,\Lambda_i\}_i,\{M_i\}_i)}{\min_{\sigma_k \in \mathcal{F}_k(\mathcal{H})} p_{\mathrm{err}}(\sigma_k,\{p_i,\Lambda_i\}_i,\{M_i\}_i)}  = \mathrm{WoR}_{\mathcal{F}_k(\mathcal{H})}(\rho). 
    \end{aligned}
\end{equation}
Taking the infimum over $k$ on both sides, we have 
\begin{equation}
    \begin{aligned}
        \inf_{k} \left\{1 -  \min_{\{p_i,\Lambda_i\}_i,\{M_i\}_i} \dfrac{p_{\mathrm{err}}(\rho,\{p_i,\Lambda_i\}_i,\{M_i\}_i)}{\min_{\sigma_k \in \mathcal{F}_k(\mathcal{H})} p_{\mathrm{err}}(\sigma_k,\{p_i,\Lambda_i\}_i,\{M_i\}_i)} \right\} = \inf_{k} \mathrm{WoR}_{\mathcal{F}_k(\mathcal{H})}(\rho).
    \end{aligned} 
\end{equation}
Since $\mathrm{WoR}_{\mathcal{F}_k(\mathcal{H})}(\rho) = \mathrm{WoR}_{\mathcal{F}(\mathcal{H})}(\rho)$, we have 
    \begin{equation}
        \begin{aligned}
            1 -  \sup_{k}\min_{\{p_i,\Lambda_i\}_i,\{M_i\}_i} \dfrac{p_{\mathrm{err}}(\rho,\{p_i,\Lambda_i\}_i,\{M_i\}_i)}{\min_{\sigma_k \in \mathcal{F}_k(\mathcal{H})} p_{\mathrm{err}}(\sigma_k,\{p_i,\Lambda_i\}_i,\{M_i\}_i)} = \mathrm{WoR}_{\mathcal{F}(\mathcal{H})}(\rho). 
        \end{aligned}
    \end{equation}
as desired. 
\end{proof}

\setcounter{theorem}{0}
\setcounter{conjecture}{0}
\section{Construction of multi-copy witness for QRT of channels}
\label{app:proof_witness_characterization_channel}
In this section, we prove  Lemma~\ref{lem:witness_characterization_channel}.
We restate the lemma for readability. 

\begin{lemma}[Lemma~\ref{lem:witness_characterization_channel} in main text]
    Let $\Lambda \in \mathcal{O}(\mathcal{H}_1\to \mathcal{H}_2)$ be a quantum channel, and let $s > 0$. 
    Consider a quantum channel $\Xi \in \mathcal{O}(\mathcal{H}_1\to \mathcal{H}_2)$.  
    Then, 
    $J_\Xi$ can be expressed as 
    \begin{equation}
        \label{eq:xi_characterization}
        J_\Xi = \frac{J_\Lambda + s' J_\Theta}{1 + s'}
    \end{equation}
    by using some positive number $0 < s' \leq s$ and some channel $\Theta \in \mathcal{O}(\mathcal{H}_1\to \mathcal{H}_2)$ 
    if and only if 
    \begin{equation}
        \label{eq:positivity_condition_channel}
        S_m\left(\frac{1+s}{s} J_\Xi - \frac{1}{s}J_\Lambda \right) \geq 0
    \end{equation}
    for all $m = 1,2,\ldots,d_1d_2$. 
\end{lemma}

\begin{proof}
    By Lemma~\ref{lem:bloch_ball}, the condition in Eq.~\eqref{eq:positivity_condition_channel} is equivalent to positivity of Hermitian operator
    \begin{equation}
        J \coloneqq \frac{1+s}{s} J_\Xi - \frac{1}{s}J_\Lambda. 
    \end{equation}
    
    First, suppose that $J_\Xi$ cannot be expressed as 
    \begin{equation}
        J_\Xi = \frac{J_\Lambda + s' \tau}{1 + s'}
    \end{equation}
    with $0 < s' \leq s$ and $\Theta \in \mathcal{O}(\mathcal{H}_1\to \mathcal{H}_2)$.
    By way of contradiction, suppose that 
    \begin{equation}
        J = \frac{1+s}{s} J_\Xi - \frac{1}{s}J_\Lambda 
    \end{equation}
    is positive. 
    Since 
    \begin{equation}
        \tr_{\mathcal{H}_2}[J] = \frac{1+s}{s} \tr_{\mathcal{H}_2}[J_\Xi] - \frac{1}{s}\tr_{\mathcal{H}_2}[J_\Lambda] = \frac{1+s}{s}I_{d_1} - \frac{1}{s}I_{d_1} = I_{d_1},  
    \end{equation}
    $J$ is a valid Choi operator. 
    In addition, we have 
    \begin{equation}
        J_\Xi = \frac{J_\Lambda + sJ_\Theta}{1 + s}. 
    \end{equation}
    However, this contradicts the fact that $J_\Xi$ cannot be written in the form in Eq.~\eqref{eq:xi_characterization} using $s' \leq s$. 
    Thus, in this case, $J$ is not positive, and thus there exists $2 \leq m\leq d$ such that 
    \begin{equation}
        S_m\left(\frac{1+s}{s} J_\Xi - \frac{1}{s}J_\Lambda \right) < 0.
    \end{equation}
    
    On the other hand, suppose that we have 
    \begin{equation}
        J_\Xi = \frac{J_\Lambda + s' J_\Theta}{1 + s'}
    \end{equation}
    with $0 < s' \leq s$ and $\Theta \in \mathcal{O}(\mathcal{H}_1\to \mathcal{H}_2)$.
    Then,  
    \begin{align}
        J
        &= \frac{1+s}{s} J_\Xi - \frac{1}{s}J_\Lambda \\ 
        &= \frac{1+s}{s}\left(\frac{J_\Lambda + s' J_\Theta}{1 + s'}\right) - \frac{1}{s}\Lambda \\ 
        &= \left(\frac{1 + s}{1 + s'} - 1\right)\frac{1}{s} J_\Lambda + \frac{(1+s)s'}{s(1+s')}J_\Theta. 
    \end{align}
     Since $0 < s'\leq s$, we have $(1 + s)/(1 + s') \geq 1$, so $J$ is positive. 
     Therefore, 
     \begin{equation}
        S_m\left(\frac{1+s}{s} \eta - \frac{1}{s}\rho \right) \geq 0
    \end{equation}
    for all $m = 1,2,\ldots,d$. 
\end{proof}

\setcounter{theorem}{0}
\setcounter{conjecture}{0}
\section{Worst-case advantage of channels in state discrimination}
\label{app:proof_worst_case_channel}
In this section, we prove Theorem~\ref{thm:advantage worst case channel}.
We restate the theorem for readability. 

\begin{theorem}[Theorem~\ref{thm:advantage worst case channel} in main text]
    For any resource channel $\Lambda$, 
    \begin{equation}
            \inf_{k} \max_{\{p_i,\eta_i\}_i,\{M_i\}_i} \dfrac{p_{\mathrm{succ}}(\Lambda,\{p_i,\eta_i\}_i,\{M_i\}_i)}{\max_{\Xi_k \in \mathcal{O}_{\textup{F}_k}} p_{\mathrm{succ}}(\Xi_k,\{p_i,\eta_i\}_i,\{M_i\}_i)} = 1 + R_{\mathcal{O}_{\textup{F}}}(\Lambda). 
    \end{equation}
\end{theorem}

\begin{proof}
For any channel $\Lambda$, 
\begin{equation}
    R_{\mathcal{O}_{\textup{F}}}(\Lambda) = \inf_{k} R_{\mathcal{O}_{\textup{F}_k}}(\rho). 
\end{equation}

For any $k$, since $\mathcal{O}_{\textup{F}_k}$ is convex, we can apply the analysis from Ref.~\cite{Takagi2019a} and get 
\begin{equation}
    \begin{aligned}
        \max_{\{p_i,\eta_i\}_i,\{M_i\}_i} \dfrac{p_{\mathrm{succ}}(\Lambda,\{p_i,\eta_i\}_i,\{M_i\}_i)}{\max_{\Xi_k \in \mathcal{O}_{\textup{F}_k}} p_{\mathrm{succ}}(\Xi_k,\{p_i,\eta_i\}_i,\{M_i\}_i)}  = 1 + R_{\mathcal{O}_{\textup{F}_k}}(\Lambda). 
    \end{aligned}
\end{equation}
Taking the infimum over $k$ on both sides, we have 
\begin{equation}
    \begin{aligned}
        \inf_{k} \max_{\{p_i,\eta_i\}_i,\{M_i\}_i} \dfrac{p_{\mathrm{succ}}(\Lambda,\{p_i,\eta_i\}_i,\{M_i\}_i)}{\max_{\Xi_k \in \mathcal{O}_{\textup{F}_k}} p_{\mathrm{succ}}(\Xi_k,\{p_i,\eta_i\}_i,\{M_i\}_i)}  = 1 + \inf_k R_{\mathcal{O}_{\textup{F}_k}}(\Lambda).
    \end{aligned} 
\end{equation}
Similarly to Lemma~\ref{lem:worst_case_robustness}, we can also show that $\inf_{k} R_{\mathcal{O}_{\textup{F}_k}}(\Lambda) = R_{\mathcal{O}_{\textup{F}}}(\Lambda)$. 
Hence, 
    \begin{equation}
        \begin{aligned}
            \inf_{k} \max_{\{p_i,\eta_i\}_i,\{M_i\}_i} \dfrac{p_{\mathrm{succ}}(\Lambda,\{p_i,\eta_i\}_i,\{M_i\}_i)}{\max_{\Xi_k \in \mathcal{O}_{\textup{F}_k}} p_{\mathrm{succ}}(\Xi_k,\{p_i,\eta_i\}_i,\{M_i\}_i)}  = 1 + R_{\mathcal{O}_{\textup{F}}}(\Lambda). 
        \end{aligned}
    \end{equation}
as desired. 
\end{proof}

\end{document}